\documentclass[twocolumn]{autart}

\usepackage[usenames]{color}
\usepackage{enumerate}
\usepackage{url}
\usepackage{subfigure}
\usepackage{amsfonts,mathrsfs}
\usepackage{amssymb,amsmath}
\usepackage{verbatim}
\usepackage{acronym}
\usepackage{mathtools}
\usepackage{graphicx}
\usepackage{latexsym}

\def\fskip#1{}

\newtheorem{theorem}{Theorem}

\newtheorem{definition}{Definition}
\newtheorem{example}{Example}

\newtheorem{lemma}{Lemma}

\newtheorem{remark}{Remark}

\renewcommand{\theenumi}{\arabic{enumi}}

\def\1{{\bf 1}}

\newcommand{\remove}[1]{}

\begin{document}

\begin{frontmatter}
\vspace{-1cm}
\title{Smart Routing of Electric Vehicles for Load Balancing in Smart Grids} 
     
     \thanks{S. Rasoul Etesami is with Department of Industrial and Systems Engineering, University of Illinois at Urbana-Champaign, (email: etesami1@illinois.edu).\\
     Walid Saad is with Wireless@VT, Department of ECE, Virginia Tech (email: walids@vt.edu).\\
     Narayan Mandayam is with WINLAB, Department of ECE, Rutgers University,
(email: narayan@winlab.rutgers.edu).\\
H. Vincent Poor is with Department of Electrical Engineering, Princeton University, (email: poor@princeton.edu).}
\thanks{This research was supported by the NSF under Grants  ECCS-1549881, ECCS-1549894, ECCE-1549900, IIS-1633363, and OAC-1541105.}                                           


\vspace{-0.7cm}
\author{S. Rasoul Etesami}, 
\author{Walid Saad}, 
\author{Narayan Mandayam},
\author{H. Vincent Poor}


\begin{abstract} 
Electric vehicles (EVs) are expected to be a major component of the smart grid. The rapid proliferation of EVs will introduce an unprecedented load on the existing electric grid due to the charging/discharging behavior of the EVs, thus motivating the need for novel approaches for routing EVs across the grid. In this paper, a novel distributed control framework based on noncooperative game theory for routing of EVs within the smart grid is proposed. The goal of this framework is to control and balance the electricity load in a distributed manner across the grid while taking into account the traffic congestion and the waiting time at charging stations. The EV routing problem is formulated as a repeated game, and it is shown that the selfish behavior of EVs will result in a pure-strategy Nash equilibrium with the price of anarchy upper bounded by the ratio of the variance of the ground load to the total number of EVs in the grid. In particular, it is shown that any achieved Nash equilibrium substantially improves the load balance across the grid. Moreover, the results are extended to capture the stochastic nature of induced ground load as well as the subjective behavior of the EV owners using the behavioral framework of prospect theory. Simulation results provide new insights on efficient energy pricing at charging stations and under realistic grid conditions.  
\end{abstract}

\end{frontmatter}


\section{Introduction}

Electric vehicles (EVs) are rapidly becoming a major component of cities around the world. Based on Bloomberg New Energy Finance, EVs are expected to represent 35 percent of new car sales globally by 2040. Greentech Media Research expects at least $11.4$ million electric vehicles (EVs) on the road only in the U.S. in 2025. Due to this rapid proliferation of EVs, an important challenge is to effectively manage and control their integration within the electric power grid \cite{rigas2015managing}. For instance, if too many EVs simultaneously charge their batteries at a charging station, it will substantially increase the load at that station, which, in turn, will be detrimental to other grid components. However, intelligently routing EVs can turn this challenge into an opportunity by viewing EVs as mobile storage devices that charge/discharge their batteries at charging stations that have extra/shortage of energy to offer for sale. This, in turn, requires introducing an appropriate mechanism design that aligns EVs' needs with the needs of the power grid.

As more EVs join the grid, the waiting time in actual road traffic and at charging stations will constitute a major problem. Since EVs need to be charged more often than fossil-fueled vehicles \cite{kelly2012time}, if there does not exist enough charging stations, we may expect long queues at the charging stations that can directly impact the comfort of EV owners. One way of handling this issue from the system level is to build additional charging stations to match the supply and demands. However, this is not the most cost-effective solution, and yet, it does not eliminate the necessity of dynamic load balancing at charging stations (e.g., due to dynamic shift of demands over time). An alternative solution to this issue is to take advantage of the distributed nature of the power grid to dynamically match supply and demands, and this is the approach that we consider in this paper. Our solutions provide a novel decentralized game-theoretic approach to the control of EVs in smart grids which captures the effect of agents' selfishness on the system outcome and its efficiency. In particular, we provide a distributed scheduling of EVs which not only balances the distribution of the electricity load but also takes into account the traffic congestion and waiting time at charging stations.    

{\bf Related Work:} There have been several recent works that investigated the challenges of managing EVs in the smart grid. In \cite{wu2012vehicle}, the authors propose a vehicle-to-aggregator interaction game and develop a pricing policy and design a mechanism to achieve optimal frequency regulation performance. The works in \cite{robu2011online} and \cite{gerding2011online} propose truthful online auction mechanisms in which agents represent EV owners who bid for energy units and also time slots in which an EV is available for charging/discharging. Similarly, the work in \cite{stein2012model} considers a consensus based online mechanism design for EV charging with pre-commitment.

A real-time traffic routing system based on an incentive compatible mechanism design has been considered in \cite{bui2012dynamic}. In this system a passenger first reports his maximum accepted travel time, and the mechanism then assigns a path that matches the passenger's preference given the current traffic conditions. In \cite{ibars2010distributed} and \cite{lutati2014congestion}, the authors propose a congestion game model to control the power demand at peak hours, by using dynamic pricing. A similar approach based on congestion games is proposed in \cite{beaude2012charging} for EV charging. A survey on utilizing artificial intelligence techniques to manage EVs over the power grid can be found in \cite{rigas2015managing}. In \cite{alizadeh2016optimal} the authors consider a coupled power and transportation network and provide an optimal pricing scheme to manage EVs over the network. However, unlike our game-theoretic framework, the approach in \cite{alizadeh2016optimal} is based on an individual optimization over an extended network. While the earlier literature provides important analytic results for managing EVs in the grid, these works mainly focus on one aspect of smart grid, (e.g., reducing the peak hour demand) without taking into account other important factors such as traffic congestion or waiting time at charging stations which are also crucial in affecting EVs' decisions.

Meanwhile, there is a rich literature on routing games where the traffic congestion is selfishly controlled by vehicle owners who seek to minimize their travel costs \cite{roughgarden2007routing,koutsoupias1999worst,roughgarden2002bad,awerbuch2005price,suri2004selfish}. Depending on whether the traffic flow can be divided among different paths one can distinguish \emph{unsplitable} and \emph{splitable} routing games \cite{awerbuch2005price}. Moreover, whether each user's contribution to the overall traffic is negligible or not one can distinguish \emph{non-atomic} and \emph{atomic} routing games \cite{roughgarden2007routing}. In this regard, one of the widely used metrics in the literature which measures efficiency and the extent to which a system degrades due to selfish behavior of its agents is the \emph{price
of anarchy} (PoA) \cite{koutsoupias1999worst}. It has been shown in \cite{roughgarden2002bad} that, for a linear latency function, the PoA of a nonatomic routing game is exactly $\frac{4}{3}$. This result has been extended later in \cite{awerbuch2005price} to splittable routing game with a slightly different bound on the PoA. Similarly, the authors in \cite{suri2004selfish} have studied the PoA of selfish load balancing in atomic congestion games. Moreover, the PoA of noncooperative demand-response in smart grids with flexible loads/EVs has been studied in \cite{chakraborty2014demand} and \cite{chakraborty2013flexible}. Recently, in  \cite{roughgarden2015intrinsic,roughgarden2015local,meir2015playing}, a so-called ``smoothness" condition has been developed under which one can obtain simple bounds on the PoA for a large class of congestion games. However, smoothness requires decoupling in arguments of the social cost function which is not immediately applicable to our model.

 Moreover, there is strong evidence \cite{kahneman1979prospect} that
real-world, human decision makers do not make decisions based on expected values of outcomes, but rather based on their perception on the potential value of losses and gains associated with an outcome. Since EVs are owned and operated by humans, the subjective perceptions and decisions of these human owners can substantially affect the grid outcomes. This makes \emph{prospect theory} (PT) \cite{kahneman1979prospect} a powerful framework that allows modeling real-life human choices, a natural choice for modeling EVs' decision making in smart grids under real behavioral considerations. Applications of PT for energy management by modifying consumer’s electricity demands have been addressed earlier in \cite{saad2016toward} and \cite{wang2016load}. However, these works do not capture the real-life decision making processes involved in the management of EVs in the smart grid. For other relevant alternative approachs (other than PT) to study risk, uncertainty, and behavioral decisions, we refer to \cite{meir2015playing} and \cite{nikolova2014mean}.

{\bf Contributions and Organization:} To address the aforementioned challenges, the main contribution of this paper is to develop a comprehensive distributed control framework for EV management in smart grids which takes into account the traffic congestion costs, the electricity price and availability, the distributed nature of the system, and the selfishness or subjective perceptions of the EV owners. Our work differs from prior art in several aspects: 1) It models the interactions between EV using a routing game \cite{roughgarden2007routing}, by taking into account the traffic congestion costs, 2) Factors in the waiting time of EVs at charging stations, 3) Introduces an energy pricing scheme to control and balance the EV load across the grid, and 4) Incorporates real-life decision behavior of EVs under uncertain energy availability by using PT and studies its deviations from conventional classical game theory (CGT). Our work is motivated by the fact that EVs can be viewed as dynamic storage devices which can move around the grid and balance the load across it. This mandates careful grid designs (e.g., pricing electricity properly at charging stations) that can align the energy needs of selfish EVs with those of the smart grid. This approach can potentially be applied to control other multi-agent network systems where the system authority has limited direct control on the agents' decisions, and yet it wants to design a mechanism in the system level to control the agents toward a certain objective (e.g., load balancing).

In the studied model, we consider a set of EVs that are traveling from an origin to a destination. Each EV may or may not stop at one of the charging stations along with its origin-destination path to charge/discharge its battery. Moreover, once joining a station, an EV can decide on the amount of energy to charge/discharge at that station. Here, the energy price charged at each station for buying or selling depends on the total energy demand at that station, a station-specific pricing function, as well as the ground load which is induced by other grid components such as residential or industrial users. Therefore, each EV chooses a route, a charging station along that route, and the amount of energy to charge/discharge at that station. We formulate the interactions between EVs as a repeated noncooperative game in which each EV seeks to minimize the tradeoff between travel time and energy price. We show that such a game admits a pure-strategy Nash equilibrium (NE) and we show that the PoA of this NE is upper bounded by the ratio of the variance of the ground load to the total number of EVs in the grid. Hence, for a large number of EVs, although each EV selfishly and independently minimizes its own cost, the social cost of all EVs will still be close to its optimal value, i.e., when a central grid authority optimally manages all the EVs. Furthermore, we show that any NE achieved as a result of the EVs' interactions will indeed improve the load balancing across the grid. We then take into account the uncertainty of the ground load and provide a bound on the number of EVs which guarantees a low PoA with high probability. In particular, we extend our model by incorporating the subjective behavior of EVs and study its deviations from CGT. Our simulation results provide new insights on energy pricing at different stations to keep the overall performance of the grid, which is measured in terms of the social cost, close to its optimal under more realistic scenarios.

The paper is organized as follows. In Section \ref{sec:model}, we introduce our system model. In Section \ref{sec:NE-PoA}, we establish the existence of pure NE points. We analyze the efficiency of NE points in terms of social cost and load balancement in Section \ref{sec:PoA}. We extend our results to a stochastic setting with PT in Section \ref{sec:stochatic-ground}. Simulation results are given in Section \ref{sec:simulation}, and conclusions are drawn in Section \ref{sec:conclusion}.

\section{System Model and Problem Formulation}\label{sec:model}

Consider a traffic network modeled as a directed graph $\mathcal{G}=(\mathcal{V},\mathcal{E})$, where each node in $\mathcal{V}$ represents a traffic intersection and each edge $e\in \mathcal{E}$ represents a road between two intersections. This network has a total of $n$ EVs (players) in the set $\mathcal{N}$. We let $n_e\in \mathbb{Z}^{\ge 0}$ be the total number of EVs on road $e$. We denote the level of battery charge of vehicle $i$ by $b_i\in [\underline{b}_i, \bar{b}_i]$, where $\underline{b}_i$ and $\bar{b}_i$ denote, respectively, the minimum level of battery charge for EV $i$ to operate, and the maximum capacity of EV $i$'s battery (note that $0<\underline{b}_i<\bar{b}_i$). In this network, we have a total of $m$ charging stations in the set $\mathcal{M}$ that are located over possibly different roads of the network. Each charging station $j\in \mathcal{M}$ can serve its EVs with a rate of $\sigma_j>0$.\footnote{$\sigma_j$ is the number of served EVs per time unit at station $j$.} We denote the set of all EVs associated to station $j$ by $\mathcal{Q}_j$. We assume that each station $j$ measures its excess/shortage energy with respect to an internal nominal reference point. However, due to malfunctioning of the operating grid, stochasticity of the generated solar/wind energy at station $j$, or other uncertain loads which are induced by nearby components, we denote the difference between the current energy level at station $j$ and its nominal reference point by $g_j \in \mathbb{R}$. Therefore, $g_j>0$ means that station $j$ is willing to offer its excess energy for sale while $g_j<0$ means that station $j$ demands for extra energy. Note that ideally station $j$ wants to have $g_j=0$ in order to keep its current energy level equal to its nominal value.

\begin{figure}[t!]
\vspace{-0.75cm}
\begin{center}
\includegraphics[totalheight=.175\textheight,
width=.275\textwidth,viewport=160 0 700 525]{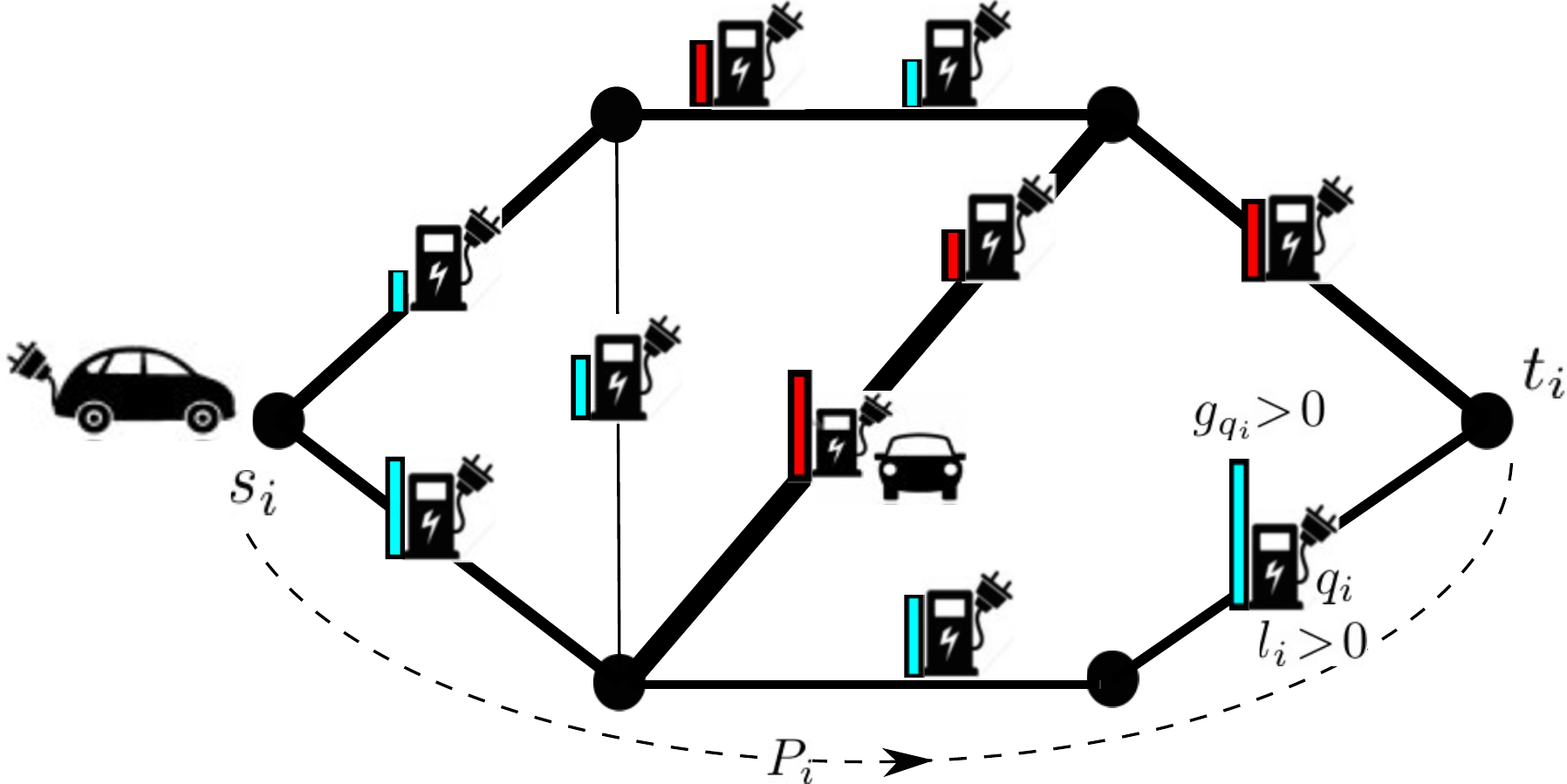} \hspace{0.4in}
\end{center}\vspace{-0.1cm}
\caption{An illustrative example of the studied model. Each EV wants to move from its origin $s_i$ to its destination $t_i$. The traffic load on each road is captured by the thickness of that edge (the thicker an edge, the more traffic on that road). The blue bar next to each station shows that the station has extra energy to offer while the red bar shows that the station is operating below its nominal value. Given the current state of the network, it seems most reasonable for EV $i$ to choose the route $P_i$ and stop by station $q_i$ to charge $l_i$ energy units.}
\label{fig:EV}
\end{figure}

We assume that each EV wants to go from its current location $s_i\in \mathcal{V}$ to its destination $t_i\in \mathcal{V}$ over a path (route) $P_i$. During this route, it can choose to charge/discharge its battery by some amount $l_i\in [\underline{b}_i-b_i, \bar{b}_i-b_i]$, at some intermediate station $q_i\in \mathcal{M}$ along that route.\footnote{See Remark \ref{rem:0-station} for the case in which an EV decides not to join any station.} Here, $l_i>0$ means that EV $i$ charges its battery by $l_i$ units of energy, while $l_i<0$ means it discharges its battery. Therefore, we can denote the action of an EV (player) $i$ by $\boldsymbol{a}_i:=(P_i, q_i,l_i)$, where $P_i$ is the path chosen by player $i$ from its source to its destination, $q_i$ is the selected charging station along $P_i$, and $l_i$ is the amount of electricity that player $i$ decides to charge or discharge at station $q_i$ (Figure \ref{fig:EV}). Finally, denoting the players' actions by $(\boldsymbol{a}_i,\boldsymbol{a}_{-i})$, we can define the cost of EV $i$ as: 
\begin{align}\label{eq:EV-cost-Queue}
C_i(\boldsymbol{a}_i,\boldsymbol{a}_{-i})&=\sum_{e\in P_i}c_{e}(n_e)+\frac{|\mathcal{Q}_{q_i}|}{\sigma_{q_i}}+\ln\Big(\frac{\bar{b}_i}{b_i+l_i}\Big)\cr 
&+\Big(f_{q_i}(\!\!\sum_{j\in\mathcal{Q}_{q_i}}\!\!l_j-g_{q_i})-f_{q_i}(\!\!\!\sum_{j\in\mathcal{Q}_{q_i}\setminus\{i\}}\!\!\!\!\!l_j-g_{q_i})\Big),\cr
\end{align}  
where $c_e(\cdot)$ is a latency function that captures the traffic congestion as a function of the total number of EVs over road $e\in \mathcal{E}$, and $f_{q_i}(\cdot)$ is the energy pricing function at station $q_i$ which is determined by the power grid. In \eqref{eq:EV-cost-Queue}, the first term captures the waiting cost of EV $i$ due to traffic congestion, the second term is the waiting cost for joining station $q_i$ which is proportional to the number of vehicles at station $q_i$, and the third term is the risk of having an empty battery which grows quickly as the battery level decreases.\footnote{Here, the choice of a logarithmic function is one way of modeling this risk which is mainly motivated by the log barrier function frequently used in convex optimization \cite{boyd2004convex}.} Finally, the last term in \eqref{eq:EV-cost-Queue} is the energy expense/income for choosing to charge/discharge $l_i$ units of electricity at station $q_i$. In this formulation, the energy price for EV $i$ equals its marginal energy contribution to station $q_i$ (see, Remark \ref{rem:1-margin}). Note that the last term in \eqref{eq:EV-cost-Queue} can also be negative, which means that EV $i$ can be paid by the system depending on the aggregate load of EVs and ground energy in station $q_i$. This incentivizes EVs who have extra energy in their batteries to join station $q_i$ and discharge their batteries thus balancing the load at that station. The first two terms in \eqref{eq:EV-cost-Queue} are in the form of delay cost while the last two terms are in terms of energy cost. However, we are implicitly assuming that these two costs can be translated to each other using a tradeoff parameter which is already absorbed in the latency functions and processing rates.  

\vspace{-0.1cm}
As it can be seen from the definition of EVs' cost functions \eqref{eq:EV-cost-Queue}, the incurred cost by an EV depends not only on its own action, but also on the other EVs' decisions. This naturally defines a noncooperative game among the EVs having the following key components: A set $\mathcal{N}$ of EVs (players). Each player $i\in \mathcal{N}$ has an action set $\mathcal{A}_i:=\mathcal{P}_i\times\mathcal{S}_i\times[\underline{b}_i-b_i,\bar{b}_i-b_i]$, where $\mathcal{P}_i$ is the set of all paths between $s_i$ to $t_i$, and $\mathcal{S}_i$ is the set of all stations along the chosen path by player $i$. Each player $i\in \mathcal{N}$ takes an action $\boldsymbol{a_i}\in \mathcal{A}_i$ and incurs a cost $C_i(\boldsymbol{a}_i,\boldsymbol{a}_{-i})$ given by \eqref{eq:EV-cost-Queue}. In this game, each EV in the grid seeks to select an action which minimizes its own cost. 

\vspace{-0.1cm}
\begin{remark}\label{rem:0-station}
\normalfont  The cost function given in \eqref{eq:EV-cost-Queue} is fairly general and can incorporate additional constraints into the model.  For instance, a situation in which some EVs prefer not to join any station (e.g. due to charging at home or workplace) can be handled by adding to each road $e\in\mathcal{E}$ a \emph{virtual} station $j$ (i.e., a station which physically does not exist, and it is only for the sake of analysis). We let all the virtual stations have an infinite speed of $\sigma_{j}=\infty$ and zero pricing function $f_j=0$. As a result, each EV $i$ has the option of joining an \emph{actual} station, in which case everything remains as before, or it will join a virtual station which translates to saying that EV $i$ will not to join any actual station. Therefore, all the results will continue to hold for this new setting except that we now have $m+|\mathcal{E}|$ stations. In particular, this model can capture the effect of non-EVs (non-EVs can be viewed as EVs that decide not to join any station and hence only incur/contribute to the traffic congestion cost).
\end{remark}

\vspace{-0.1cm}
\begin{remark}\label{rem:1-margin}
\normalfont The rationale behind using marginal pricing is that when the load in a station is high (e.g., due to EV congestion in that station), marginal pricing becomes effective and sets a higher price in that station. This disincentivizes more EVs to join that station. It is worth noting that the use of marginal pricing is not specific to our work and has been extensively justified in economics \cite{gerding2011online,sandholm2007pigouvian}, modeling of EVs \cite{gerding2011online,bui2012dynamic}, and engineering applications \cite{meir2016marginal,bui2012dynamic}. For instance, \cite{gerding2011online} and \cite{bui2012dynamic} utilize marginal payment strategy to design truthful mechanisms for EV charging. We note that many pricing policies can be implemented as a special case of marginal pricing. For instance, a fixed pricing policy which charges an EV a constant amount of $c$ per unit of electricity usage can be implemented using the linear pricing function $f(x)=cx$.     
\end{remark}

\vspace{-0.1cm}
\begin{example}
\normalfont Given a road $e\in\mathcal{E}$, a natural choice for the latency function is the linear latency $c_e(x):=a_ex+b_e$, where $b_e$ can be thought as the length of that road. This means that the travel time of a vehicle that chooses road $e$ depends on the length of that road and linearly increases in terms of the number of other vehicles on that road. In particular, we may assume that the electricity cost of traveling over road $e$ is implicitly captured into this cost function. Otherwise, if an EV incurs $a'_ex+b'_e$ amount of electricity cost due to travel on road $e$ with congestion $x$, then by defining $c_e(x):=(a_e+\lambda a'_e)x+(b_e+\lambda a'_e)$ we can capture both delay and electricity cost with $\lambda$ being the energy-delay tradeoff parameter.   
\end{example}

\vspace{-0.2cm} 
\section{Existence of Pure Nash Equilibrium}\label{sec:NE-PoA}   

\vspace{-0.1cm}
Our first goal is to see whether the EVs' game will yield a stable outcome, as captured by the notion of a NE:

\begin{definition}
An action profile $(\boldsymbol{a}_i,\boldsymbol{a}_{-i})$ is called a pure-strategy Nash equilibrium (NE) for the EVs' interaction game if $C_i(\boldsymbol{a}_i,\boldsymbol{a}_{-i})\leq C_i(\boldsymbol{a}'_i,\boldsymbol{a}_{-i}), \forall i\in\mathcal{N}$ and $\boldsymbol{a}'_i\in \mathcal{A}_i$.
\end{definition}
Next, we show that the EVs interaction game admits a pure-strategy NE, meaning that although each EV aims to minimize its own cost, they collectively will converge to a stable outcome where every EV is satisfied as long as others do not deviate.     
\begin{theorem}\label{thm:existence-pure}
The EVs' game admits a pure-strategy NE.
\end{theorem}
\begin{proof}
We show that the EVs' game is a potential game, and hence, it admits a pure NE. Let
\begin{align}\label{eq:potential}
\!\!\!\!\!\!\!\!\Phi(\boldsymbol{a}_i,\boldsymbol{a}_{-i})&:=\underbrace{\sum_{e\in \mathcal{E}}\sum_{x=1}^{n_e}c_{e}(x)}_\text{$\phi_1(a_i,a_{-i})$}+\underbrace{\sum_{\ell=1}^{m}\frac{|\mathcal{Q}_{\ell}|(|\mathcal{Q}_{\ell}|+1)}{2\sigma_{\ell}}}_\text{$\phi_2(a_i,a_{-i})$}\cr 
&+\underbrace{\sum_{\ell=1}^{m}f_{\ell}(\sum_{j\in\mathcal{Q}_{\ell}}l_{j}-g_{\ell})+\sum_{j=1}^{n}\ln\left(\frac{\bar{b}_j}{b_j+l_j}\right)}_\text{$\phi_3(a_i,a_{-i})$}\!.
\end{align}
We will show that for any two actions $\boldsymbol{a}_i=(P_i,q_i,l_i)$ and $\boldsymbol{a}'_i=(P'_i,q'_i,l'_i)$, we have $\Phi(\boldsymbol{a}_i,\boldsymbol{a}_{-i})-\Phi(\boldsymbol{a}'_i,\boldsymbol{a}_{-i})=C_i(\boldsymbol{a}_i,\boldsymbol{a}_{-i})-C_i(\boldsymbol{a}'_i,\boldsymbol{a}_{-i})$. To show this, first we note that the traffic congestion cost in $C_i(\boldsymbol{a}'_i,\boldsymbol{a}_{-i})$ is equal to $\sum_{e\in P'_i\cap P_i}c_{e}(n_e)+\sum_{e\in P'_i\setminus P_i}c_{e}(n_e+1)$. This is because if EV $i$ changes its path from $P_i$ to $P'_i$, then the number of vehicles $n_e$ in all the roads $e\in P_i\cap P'_i$ remains as before. However, the number of vehicles in roads $e\in P'_i\setminus P_i$ increases by exactly 1 (as now vehicle $i$ has joined these roads). Similarly, if vehicle $i$ leaves station $q_i$ to join station $q'_i$, its new waiting cost will change from $\frac{|\mathcal{Q}_{q_i}|}{\sigma_{q_i}}$ to $\frac{|\mathcal{Q}_{q'_i}\setminus\{i\}|+1}{\sigma_{q'_i}}$. Following the same argument for the cost associated with the marginal energy price, we can write     
\begin{align}\nonumber
C_i(\boldsymbol{a}'_i,\boldsymbol{a}_{-i})&=\sum_{e\in P'_i\cap P_i}c_{e}(n_e)+\sum_{e\in P'_i\setminus P_i}c_{e}(n_e+1)\cr 
&+\frac{|\mathcal{Q}_{q'_i}\!\setminus\!\{i\}|+1}{\sigma_{q'_i}}+\ln\Big(\frac{\bar{b}_i}{b_i+l'_i}\Big)\cr 
&+f_{q'_i}(\!\!\!\!\sum_{j\in\mathcal{Q}_{q'_i}\!\setminus\!\{i\}}\!\!\!\!\!l_j+l'_i-g_{q'_i})\!-\!f_{q'_i}(\!\!\!\!\sum_{j\in\mathcal{Q}_{q'_i}\!\setminus\!\{i\}}\!\!\!\!\!l_j-g_{q'_i}).
\end{align} 
By subtracting $C_i(\boldsymbol{a}'_i,\boldsymbol{a}_{-i})$ from the cost function $C_i(\boldsymbol{a}_i,\boldsymbol{a}_{-i})$ given in \eqref{eq:EV-cost-Queue}, we obtain
\begin{align}\nonumber
C_i(\boldsymbol{a}_i,\boldsymbol{a}_{-i})&-C_i(\boldsymbol{a}'_i,\boldsymbol{a}_{-i})\cr 
&=\!\!\!\sum_{e\in P_i\setminus P'_i}\!\!\!c_e(n_e)-\!\!\!\sum_{e\in P'_i\setminus P_i}\!c_e(n_e\!+\!1)\cr 
&+\frac{|\mathcal{Q}_{q_i}|}{\sigma_{q_i}}-\frac{|\mathcal{Q}_{q'_i}\!\setminus\!\{i\}|+1}{\sigma_{q'_i}}+\ln\Big(\frac{b_i+l'_i}{b_i+l_i}\Big)\cr 
&+\Big(f_{q_i}(\!\!\sum_{j\in\mathcal{Q}_{q_i}}\!\!\!l_j-g_{q_i})\!-\!f_{q_i}(\!\!\!\sum_{j\in\mathcal{Q}_{q_i}\setminus\{i\}}\!\!\!l_j-g_{q_i})\Big)\cr 
&-\Big(f_{q'_i}(\!\!\!\!\!\sum_{j\in\mathcal{Q}_{q'_i}\!\setminus\!\{i\}}\!\!\!\!\!l_j+l'_i-g_{q'_i})\!-\!f_{q'_i}(\!\!\!\!\!\sum_{j\in\mathcal{Q}_{q'_i}\!\setminus\!\{i\}}\!\!\!\!\!l_j-g_{q'_i})\Big).
\end{align}
Next we consider the change in the potential function due to an action change of player $i$. We can write:
\begin{align}\nonumber
&\phi_1(\boldsymbol{a}_i,\boldsymbol{a}_{-i})-\phi_1(\boldsymbol{a}'_i,\boldsymbol{a}_{-i})=\!\!\!\sum_{e\in P_i\setminus P'_i}\!\!\!c_e(n_e)-\!\!\!\!\!\sum_{e\in P'_i\setminus P_i}\!\!\!c_e(n_e\!+\!1),\cr 
&\phi_2(\boldsymbol{a}_i,\boldsymbol{a}_{-i})-\phi_2(\boldsymbol{a}'_i,\boldsymbol{a}_{-i})=\frac{|\mathcal{Q}_{q_i}|}{\sigma_{q_i}}-\frac{|\mathcal{Q}_{q'_i}\!\setminus\!\{i\}|+1}{\sigma_{q'_i}},\cr 
&\phi_3(\boldsymbol{a}_i,\boldsymbol{a}_{-i})-\phi_3(\boldsymbol{a}'_i,\boldsymbol{a}_{-i})=\ln\Big(\frac{\bar{b}_i}{b_i+l_i}\Big)-\ln\Big(\frac{\bar{b}_i}{b_i+l'_i}\Big)\cr 
&\qquad\qquad+\Big(f_{q_i}(\!\sum_{j\in\mathcal{Q}_{q_i}}\!\!l_j-g_{q_i})+f_{q'_i}(\!\!\!\sum_{j\in\mathcal{Q}_{q'_i}\!\setminus\!\{i\}}\!\!\!\!\!l_j-g_{q'_i})\Big)\cr 
&\qquad\qquad-\Big(f_{q_i}(\!\!\!\!\!\sum_{j\in\mathcal{Q}_{q_i}\!\setminus\!\{i\}}\!\!\!\!\!l_j-g_{q_i})+f_{q'_i}(\!\!\!\!\sum_{j\in\mathcal{Q}_{q'_i}\!\setminus\!\{i\}}\!\!\!\!\!l_j+l'_i-g_{q'_i})\Big).
\end{align}
Summing all the above inequalities and noting that $\Phi(\boldsymbol{a})\!-\!\Phi(\boldsymbol{a}'_i,\boldsymbol{a}_{-i})=\sum_{k=1}^{3}[\phi_k(\boldsymbol{a})\!-\!\phi_k(\boldsymbol{a}'_i,\boldsymbol{a}_{-i})]$, we get $\Phi(\boldsymbol{a})-\Phi(\boldsymbol{a}'_i,\boldsymbol{a}_{-i})=C_i(\boldsymbol{a})-C_i(\boldsymbol{a}'_i,\boldsymbol{a}_{-i})$.
\end{proof}

Theorem \ref{thm:existence-pure} shows that a \emph{pure-strategy} NE exists even though the actions of the players can take both discrete and continuous quantities or they can be highly coupled (e.g. choosing what station to join highly depends on what route to choose). Even though this theorem does not characterize uniqueness or efficiency of the equilibrium points, as we will show in Theorem \ref{thm:poa}, for a large number of EVs and specific choices of latency and pricing functions, all the equilibrium points will be almost equally efficient in terms of the social cost. In particular, we will show that for a large number of EVs, the social cost of any NE is at most a small constant factor worse than the optimal social cost. 

It is worth noting that the result of Theorem 1 is very strong in the sense that not only it guarantees the existence of a pure NE, but also it shows that \emph{any sequence} of unilaterally updates by the EVs will eventually converge to a NE. This allows us to implement the EV game as a repeated game between EV owners who will, daily, travel the distance between their home (origin) to their work (destination). The information that the EVs require to compute their optimal strategies (e.g., road congestion or charging station loads) can be broadcast using a data platform or directly can be sent to the GPS devices of the EVs. This allows each EV to have access to the most updated information of the grid state before taking its action. As a result, if EV $i$ first takes its action, this update will change the state of the entire grid whose information will be immediately available to all others. Now if a new player takes its action by best responding to the newly updated state, and this process continuous again and again, then the grid state will eventually converge to an NE (whose efficiency in terms of social cost and load balancing is established in Theorems 2 and 4).

\section{Price of Anarchy and Load Balancement}\label{sec:PoA}
In this section, we analyze the efficiency of the NE points in the EVs' game in terms of the price of anarchy (PoA) and load balancement. We first start by analyzing the price of anarchy of the EVs' game which is an important measure to capture how much the selfish behavior of the EVs can influence the overall optimality of the grid \cite{koutsoupias1999worst}.
\begin{definition}\label{def:poa}
For the EVs' interaction game, the \emph{PoA} is defined as the ratio of the maximum social cost for all Nash equilibria over the minimum (optimal) social cost, 
i.e., $\emph{PoA}=\frac{\max_{\boldsymbol{a}\in \rm NE}\sum_{i=1}^{n}C_i(\boldsymbol{a})}{\min_{\boldsymbol{a}} \sum_{i=1}^{n}C_i(\boldsymbol{a})}$.
\end{definition}
Here, optimality is measured in terms of EVs' social cost assuming that a network authority with complete information manages the EVs and seeks to minimize the overall social cost. Since EVs are selfish entities whose actions cannot be centrally controlled, modeling EVs' interactions as a game that yields a small PoA is very important. Interestingly, the following theorem shows that for linear latency and quadratic energy pricing, the PoA remains small assuming a large number of vehicles in the grid. It is worth noting that the choice of linear latency is not specific to our work only, and it has been frequently used in the game theory literature \cite{roughgarden2007routing,roughgarden2002bad}. This is because, despite its simplicity, linear latency function can still capture two main features of travel costs, namely lengths of the roads $b_e$ and traffic congestion costs $a_ex$. On the other hand, quadratic pricing is only one way of pricing to assure improved load balancement (Theorem \ref{thm:load-reduction}) while still result in a small PoA. However, to avoid case dependent analysis, in this section, we only develop our analysis for the case of linear latency and quadratic pricing functions (with a potential to be generalized in parts to more general functions). We will complete these results by providing numerical simulations in Section \ref{sec:simulation} to illustrated the tradeoff between PoA and load balancing for other choices of latency/pricing functions.

\begin{theorem}\label{thm:poa}
For linear latency $c_e(x)=a_e x+b_e$, and quadratic energy pricing function $f_j(x)=x^2$, assume that each player incurs at least a unit of cost. Then, $\emph{PoA}\leq c+4.5\big(\frac{\sum_{j=1}^{m}g^2_j}{n}\big)$, where $n$ is number of EVs, $g_j$ is the ground load at station $j$, and $c$ is a constant.   
\end{theorem}
\begin{proof}
Let $\{\boldsymbol{a}_i\!=\!(P_i,q_i,l_i)\}_{i=1}^{n}$ denote an \emph{arbitrary} Nash equilibrium, and $\{\boldsymbol{a}^*_i\!=\!(P^*_i,q^*_i,l^*_i)\}_{i=1}^{n}$ be the social optimal solution. Moreover, let $NE:=\sum_{i=1}^{n}C_i(\boldsymbol{a}_i,\boldsymbol{a}_{-i})$ and $OPT:=\sum_{i=1}^{n}C_i(\boldsymbol{a}^*_i,\boldsymbol{a}^*_{-i})$ be the social cost of this equilibrium and the optimal social cost. By definition of Nash equilibrium, for any $i\in [n]$,
\begin{align}\nonumber
C_i(\boldsymbol{a}_i,\boldsymbol{a}_{-i})&\leq C_i(\boldsymbol{a}^*_i,\boldsymbol{a}_{-i})\cr 
&=\!\!\!\sum_{e\in P^*_i\setminus P_i}\!\!\!c_e(n_e+1)+\!\!\!\sum_{e\in P^*_i\cap P_i}\!\!\!c_e(n_e)\cr 
&+\frac{|\mathcal{Q}_{q^*_i}\!\setminus\!\{i\}|+1}{\sigma_{q^*_i}}+\ln\Big(\frac{\bar{b}_i}{b_i+l^*_i}\Big)\cr 
&+\big(\!\!\!\!\!\!\sum_{k\in \mathcal{Q}_{q^*_{i}}\setminus\{i\}}\!\!\!\!\!\!l_k+l^*_i-g_{q^*_{i}}\big)^2-\big(\!\!\!\!\!\!\sum_{k\in \mathcal{Q}_{q^*_{i}}\setminus\{i\}}\!\!\!\!\!\!l_k-g_{q^*_{i}}\big)^2.
\end{align}
Summing all the above inequalities for $i\in[n]$ we obtain
\begin{align}\label{eq:three-summand}
\!\!\!\!\!\!\!\!\!NE&\leq \sum_{i=1}^{n}C_i(\boldsymbol{a}^*_i,\boldsymbol{a}_{-i})\cr 
&=\sum_{i=1}^{n}\Big(\!\sum_{e\in P^*_i\setminus P_i}\!\!\!c_e(n_e+1)+\!\!\!\sum_{e\in P^*_i\cap P_i}\!\!\!c_e(n_e)\Big)\cr 
&+\sum_{i=1}^{n}\frac{|\mathcal{Q}_{q^*_i}\!\setminus\!\{i\}|\!+\!1}{\sigma_{q^*_i}}+\sum_{i=1}^{n}\ln\Big(\frac{\bar{b}_i}{b_i+l^*_i}\Big)\cr 
&+\sum_{i=1}^{n}\!\Big[\big(\!\!\!\sum_{k\in \mathcal{Q}_{q^*_{i}}\!\setminus\!\{i\}}\!\!\!\!\!\!l_k\!+\!l^*_i\!-\!g_{q^*_{i}}\big)^2\!-\!\big(\!\!\!\!\sum_{k\in \mathcal{Q}_{q^*_{i}}\!\setminus\!\{i\}}\!\!\!\!\!l_k\!-\!g_{q^*_{i}}\big)^2\Big].
\end{align}
Next we upper bound each of the three summands in \eqref{eq:three-summand}. To this end, let $OPT_1$ and $NE_1$ denote the traffic congestion costs in the optimal solution and the NE, assuming a linear latency function $c_e(x)=a_ex+b_e$, i.e.,
\begin{align}\nonumber
&OPT_1\!:=\sum_{i=1}^{n}\sum_{e\in P^*_i}c_e(n^*_e)=\sum_{e\in\mathcal{E}}n^*_e(a_e n^*_e+b_e),\cr
&NE_1\!:=\sum_{i=1}^{n}\sum_{e\in P_i}c_e(n_e)=\sum_{e\in\mathcal{E}}n_e(a_e n_e+b_e),
\end{align}
where $n_e$ and $n^*_e$ denote the number of vehicles on edge $e\in \mathcal{E}$ due to the Nash equilibrium and the optimal solution. Using a similar method as in \cite{awerbuch2005price} we can bound the first summand in \eqref{eq:three-summand} as follows:  
\begin{align}\label{eq:first-term}
\sum_{i=1}^{n}&\Big(\!\sum_{e\in P^*_i\setminus P_i}\!\!\!c_e(n_e+1)+\!\!\!\sum_{e\in P^*_i\cap P_i}\!\!\!c_e(n_e)\Big)\cr 
&\leq \sum_{i=1}^{n}\sum_{e\in P^*_i}c_e(n_e+1)\cr 
&=\sum_{e\in \mathcal{E}}a_en^*_e n_e+\sum_{e\in \mathcal{E}}n^*_e(a_e+b_e)\cr 
&\leq \sqrt{\sum_ea_en^2_e \sum_e a_e (n_e^{*})^2}+\sum_e n^*_e(a_e n^*_e+b_e)\cr 
&\leq \sqrt{\sum_e(a_en^2_e+b_e) \sum_e (a_e (n_e^{*})^2+b_e)}\cr 
&\qquad+\sum_e n^*_e(a_e n^*_e+b_e)\cr 
&=\sqrt{NE_1\times OPT_1}+OPT_1.
\end{align}
where the first inequality is obtained by upper bounding $c_e(n_e)$ by $c_e(n_e+1)$ for each $e\in P^*_i\cap P_i$, the first equality holds by the definition of linear latency function, and the second inequality is due to Cauchy-Schwartz inequality and the fact that $n^*_e(a_e+b_e)\leq n^*_e(a_e n^*_e+b_e), \forall e$.  

To upper bound the second summand in \eqref{eq:three-summand}, let us define $OPT_2$ and $NE_2$ to be the total waiting cost at all the stations in the NE and the optimal solution, respectively,
\begin{align}\nonumber
OPT_2:=\sum_{i=1}^{n}\frac{|\mathcal{Q}^*_{q^*_i}|}{\sigma_{q^*_i}}, \ \ \ \ NE_2:=\sum_{i=1}^{n}\frac{|\mathcal{Q}_{q_i}|}{\sigma_{q_i}}.
\end{align}
Following identical argument as above in which roads $e\in\mathcal{E}$ are replaced by charging stations $j\in\mathcal{M}$, and the quantities $(n_e,n^*_e,a_e,b_e), e\in\mathcal{E}$ are replaced by $(|\mathcal{Q}_j|,|\mathcal{Q}^*_j|,\frac{1}{\sigma_j},0), j\in\mathcal{M}$, we obtain, 
\begin{align}\label{eq:second-term}
\sum_{i=1}^{n}\frac{|\mathcal{Q}_{q^*_i}\setminus\{i\}|+1}{\sigma_{q^*_i}}\leq \sqrt{OPT_2\times NE_2}+OPT_2.  
\end{align}

Finally, to bound the last summand in \eqref{eq:three-summand}, let us define $OPT_3$ and $NE_3$ to be the total energy cost induced by the optimal solution and the NE, respectively, i.e,
\begin{align}\nonumber
&OPT_3\!:=\!\sum_{i=1}^{n}\!\Big[(\!\!\sum_{k\in \mathcal{Q}^*_{q^*_{i}}}\!\!\!l^*_k\!-\!g_{q^*_{i}})^2\!-\!(\!\!\!\!\!\!\!\sum_{k\in \mathcal{Q}^*_{q^*_{i}}\!\setminus\!\{i\}}\!\!\!\!\!\!\!l^*_k\!-\!g_{q^*_{i}})^2\!+\!\ln\big(\frac{\bar{b}_i}{b_i\!+\!l^*_i}\big)\Big],\cr 
&NE_3\!:=\!\sum_{i=1}^{n}\!\Big[(\!\!\sum_{k\in \mathcal{Q}_{q_{i}}}\!\!\!l_k\!-\!g_{q_{i}})^2\!-\!(\!\!\!\!\!\sum_{k\in \mathcal{Q}_{q_{i}}\!\setminus\!\{i\}}\!\!\!\!\!\!\!l_k\!-\!g_{q_{i}})^2\!+\!\ln\big(\frac{\bar{b}_i}{b_i\!+\!l_i}\big)\Big].
\end{align}
It is shown in the Appendix (Lemma \ref{apx:thm:poa}) that,  
\begin{align}\nonumber
&\sum_{i=1}^{n}\Big[(\!\!\!\sum_{k\in \mathcal{Q}_{q^*_{i}}\!\setminus\!\{i\}}\!\!\!\!\!\!\!l_k\!+\!l^*_i\!-\!g_{q^*_{i}})^2\!-\!(\!\!\!\!\!\!\sum_{k\in \mathcal{Q}_{q^*_{i}}\!\setminus\!\{i\}}\!\!\!\!\!\!l_k\!-\!g_{q^*_{i}})^2\!+\!\ln\big(\frac{\bar{b}_i}{b_i+l^*_i}\big)\Big]\cr 
&\qquad\leq \sqrt{\left(\gamma+NE_3\right)\left(\gamma+OPT_3\right)}+OPT_3\cr 
&\qquad+\delta(\sqrt{\gamma+NE_3}+\sqrt{\gamma+OPT_3})+(\delta^2+4nb_{\max}^2),
\end{align}
where $b_{\max}\!:=\!\max_i\bar{b}_i, \delta^2\!:=\!\frac{\sum_{j=1}^{m}g^2_j}{2}$, and $\gamma\!:=\!\delta^2\!+nb^2_{\max}$. Substituting this relation together with \eqref{eq:first-term} and \eqref{eq:second-term} into \eqref{eq:three-summand}, we can write

\vspace{-0.5cm}
\begin{small}
\begin{align}\label{eq:final-POA}
&NE\leq \sqrt{OPT_1\times NE_1}+\sqrt{OPT_2\times NE_2}\cr 
&\qquad+\sqrt{\left(\gamma+NE_3\right)\left(\gamma+OPT_3\right)}+(OPT_1+OPT_2+OPT_3)\cr 
&\qquad+\delta(\sqrt{\gamma+NE_3}+\sqrt{\gamma+OPT_3})+\delta^2+4nb_{\max}^2\cr 
&\leq \sqrt{(\gamma+OPT_1+OPT_2+OPT_3)(\gamma+NE_1+NE_2+NE_3)}\cr 
&\qquad+OPT+\delta(\sqrt{\gamma+NE_3}+\sqrt{\gamma+OPT_3})+\delta^2+4nb_{\max}^2\cr 
&=\sqrt{(\gamma+OPT)(\gamma+NE)}+OPT\cr 
&\qquad+\delta(\sqrt{\gamma+NE_3}+\sqrt{\gamma+OPT_3})+\delta^2+4nb_{\max}^2\cr 
&\leq\sqrt{(\gamma+OPT)(\gamma+NE)}+OPT\cr 
&\qquad+\delta(\sqrt{\gamma+NE}+\sqrt{\gamma+OPT})+\delta^2+4nb_{\max}^2, 
\end{align}
\end{small}where the first inequality holds because for any four positive numbers $a_1,a_2,a_3,a_4$, we have $\sqrt{a_1a_2}+\sqrt{a_3a_4}\leq \sqrt{(a_1+a_3)(a_2+a_4)}$. Also, the last inequality stems from the fact that $NE_3\leq NE$, and $OPT_3\leq OPT$. Dividing both sides of \eqref{eq:final-POA} by $OPT$ (note that by the assumption $OPT\ge n$)\footnote{For instance, this assumption can be easily satisfied by charging each EV a toll of unit cost for using the system.} and setting $x=\frac{NE}{OPT}$, we get
\begin{align}\label{eq:x-poa}
x&\leq \sqrt{(\frac{\gamma}{n}+1)(\frac{\gamma}{n}+x)}+(1+\frac{\delta^2}{n}+4b_{\max}^2)\cr 
&\qquad+\frac{\delta}{\sqrt{n}}(\sqrt{\frac{\gamma}{n}+x}+\sqrt{\frac{\gamma}{n}+1}). 
\end{align}This in view of Lemma \ref{lemm:PoA-equation} (Appendix A) shows that $x\leq c+4.5\big(\frac{\sum_{j=1}^{m}g^2_j}{n}\big)$ where $c:=3+12b^2_{\max}$ is a constant.  
\end{proof}

Typically, in real grids, one can assume that each player incurs a unit cost in the system (for example we charge each EV  \$1 as a toll of using roads or other grid facilities). Then, as a result of Theorem \ref{thm:poa}, if there are many EVs in the grid (i.e., $n$ is large), although every EV minimizes its own cost, the entire grid will still operate close to its social optimal state and within only a small constant factor $c$. This allows us to align the selfish EVs' needs with those of the grid and achieve nearly the same optimal social cost when a central grid authority dictates decisions to EVs. It is worth noting that Theorem \ref{thm:poa} does not imply that, for a large number of EVs, the players' costs are less (clearly, for a higher number of EVs, the traffic congestion and waiting time at charging stations is high). However, it shows that, for a large number of EVs, there is no way to substantially reduce the aggregate cost of all the EVs more than what it is already achieved at a NE.

Next we consider a similar efficiency metric to the PoA, namely the \emph{price of stability} (PoS), which compares the social cost of the ``best" NE over the optimal cost, i.e., $\mbox{PoS}=\frac{\min_{\boldsymbol{a}\in \rm NE} \sum_{i=1}^{n} C_i(\boldsymbol{a})}{\min_{\boldsymbol{a}} \sum_{i=1}^{n} C_i(\boldsymbol{a})}$. In this case one can obtain a tighter bound for the PoS as stated below:
\begin{theorem}\label{thm:PoS}
For the linear latency and quadratic pricing function, the \emph{PoS} of the EVs' interaction game is upper bounded by $\emph{PoS}\leq 2 \big(1+b_{\max}^2+\frac{\sum_jg_j^2}{n}\big)$.
\end{theorem}
\begin{proof}
The proof can be found in Appendix \ref{appx:pos}. 
\end{proof}

Finally, we show that any NE achieved by the EVs will indeed improve the load balance in the grid. For this purpose, let us first consider the following definition: 
\begin{definition}
Let $b_{\min}=\min_i \underline{b}_i$. We refer to a station $j$ as a \emph{good} station if $|g_j|\leq \frac{\sqrt{5}}{2b_{\min}}$. Otherwise, we refer to it as a \emph{bad} station. We denote the set of bad stations by $\mathcal{B}$.
\end{definition}
Based on this definition, the load imbalance of a good station is very small and close to $0$ which eliminates the necessity of load balancing in that station. Consider the initial load \emph{imbalance} of the grid determined by the variance of the initial ground loads at all the bad stations $V_0:=\sum_{j\in \mathcal{B}}g_j^2$. Now if we let $g_j^{\rm NE}$ be the aggregate load induced by a NE at station $j$, we can express the improvement of load balancing at that NE by $V_0-V_{\rm NE}$, where $V_{\rm NE}:=\sum_{j\in\mathcal{B}}(g_j^{\rm NE})^{2}$ denotes the load variance of the bad stations at that achieved NE.   The following theorem shows that every achieved NE improves the load balance in all the bad stations without hurting the good stations.  
\begin{theorem}\label{thm:load-reduction}
Given the quadratic pricing function, let us assume for simplicity that all EVs have the same initial battery level $b_i=b, \forall i$. Then, for any NE, all the good stations will remain good while the bad stations become more balanced. In particular, $V_{\rm NE}<V_0-\sum_{j\in \mathcal{B}}\mu^2_j$, where 
\begin{align}\nonumber
\mu_j\!=\!\begin{cases} |\mathcal{Q}_j|(b_{\min}\!-\!b) & \!\!\!\mbox{if} \ \ g_j\!\leq\! (2|\mathcal{Q}_j|\!-\!1)(b_{\min}\!-\!b)\!-\!\frac{1}{2b_{\min}}, \\ 
                                  |\mathcal{Q}_j|(b_{\max}\!-\!b) & \!\!\!\mbox{if} \ \ g_j\!\ge\! (2|\mathcal{Q}_j|\!-\!1)(b_{\max}\!-\!b)\!+\!\frac{1}{2b_{\max}}, \\ 
                                  \frac{1}{2}g_j & \mbox{else},\end{cases}
\end{align}and $|\mathcal{Q}_j|$ denotes the number of EVs at station $j$ at NE. 
\end{theorem}
\begin{proof}
Consider an arbitrary but fixed station $j$. Note that for any NE $\{(P_i,q_i,l_i)\}_{i\in \mathcal{N}}$, the load components of all the players who join station $j$, i.e., $\{l_i: i\in \mathcal{Q}_j\}$ must form a NE if the players' costs are restricted to only the load portion of their costs. In other words, for every $i\in \mathcal{Q}_j$, if we consider $|\mathcal{Q}_j|$ players with costs
\begin{align}\nonumber
\hat{C}_i(l'_i,l'_{-i})\!&=\!(\!\sum_{j\in\mathcal{Q}_{j}}\!l'_j-g_{j})^2\!-\!(\!\!\!\!\!\sum_{j\in\mathcal{Q}_{j}\setminus\{i\}}\!\!\!\!l'_j-g_{j})^2\!+\!\ln\Big(\frac{\bar{b}_i}{b_i+l'_i}\Big)\cr 
&=l'^2_i+2l'_i(\!\!\!\sum_{j\in\mathcal{Q}_{j}\setminus\{i\}}\!\!\!\!l'_j-g_j)+\ln\Big(\frac{b_{\max}}{b_i+l'_i}\Big), 
\end{align} 
then, $\{l_i: i\in \mathcal{Q}_j\}$ must be a NE for this restricted game.\footnote{Note that this property only holds for a fixed charging station and not necessarily across different stations.} Since, for every $i,k\in \mathcal{Q}_j$ we have $\frac{\partial^2}{\partial^2 l'_i} \hat{C}_i=2+\frac{1}{(b+l'_i)^2}$, and $\frac{\partial^2}{\partial l'_i\partial l'_k} \hat{C}_i=2$, the restricted game with cost functions $\hat{C}_i(l'_i,l'_{-i})$ is a strictly convex game and admits a unique pure NE \cite[Theorem 2]{rosen1965existence}, that is $\{l_i: i\in \mathcal{Q}_j\}$. In addition, since by assumption $b_i=b, \forall i$, all the players have the same cost function. As a result the restricted game is a symmetric convex game which means that its unique equilibrium is symmetric \cite[Theorem 3]{cheng2004notes}. Thus $l_i=l, \forall i\in \mathcal{Q}_j$, for some $l\in [b_{\min}-b,b_{\max}-b]$. As a result, the load costs for all the players $i\in \mathcal{Q}_j$ at the NE are the same and equal to
\begin{align}\label{eq:equilibrium-symmetric-restricted_cost}
\hat{C}_i(l)=(2|\mathcal{Q}_{j}|-1)l^2-2lg_j+\ln\Big(\frac{b_{\max}}{b+l}\Big). 
\end{align} 
In particular, the equilibrium load $l$ must be the unique minimizer of \eqref{eq:equilibrium-symmetric-restricted_cost} in the feasible range $[b_{\min}-b,b_{\max}-b]$, which is given by 
\begin{align}\nonumber
l\!=\!\begin{cases} b_{\min}-b & \!\!\!\!\!\!\!\!\!\!\!\!\!\!\!\!\!\!\!\!\!\!\!\! \mbox{if} \ \ g_j\leq (2|\mathcal{Q}_j|-1)(b_{\min}-b)-\frac{1}{2b_{\min}}, \\ 
                          b_{\max}-b & \!\!\!\!\!\!\!\!\!\!\!\!\!\!\!\!\!\!\!\!\!\!\!\! \mbox{if} \ \ g_j\ge (2|\mathcal{Q}_j|-1)(b_{\max}-b)-\frac{1}{2b_{\max}}, \\ 
                           \frac{2g_j-\Psi+\sqrt{\Psi^2+4|\mathcal{Q}_j|-2}}{4|\mathcal{Q}_j|-2}        & \mbox{otherwise},\end{cases}
\end{align}where $\Psi:=(2|\mathcal{Q}_j|-1)b+g_j$. For each of the above three possibilities, we compute the equilibrium load reduction at station $j$ given by $(g_j^{\rm NE})^2-g_j^2=(|\mathcal{Q}_j|l-g_j)^2-g^2_j$.

\noindent{\bf Case I:} If $g_j\leq (2|\mathcal{Q}_j|-1)(b_{\min}-b)-\frac{1}{2b_{\min}}$, we have
\vspace{-0.5cm} 
\begin{small}
\begin{align}\nonumber
&|\mathcal{Q}_j|^2l^2-2|\mathcal{Q}_j|lg_j=(|\mathcal{Q}_j|(b_{\min}-b))^2-2|\mathcal{Q}_j|(b_{\min}-b)g_j\cr 
&\qquad\leq (-3|\mathcal{Q}_j|^2+2|\mathcal{Q}_j|)(b_{\min}-b)^2+\frac{|\mathcal{Q}_j|(b_{\min}-b)}{b_{\min}}\cr 
&\qquad\leq (-3|\mathcal{Q}_j|^2+2|\mathcal{Q}_j|)(b_{\min}-b)^2\leq -|\mathcal{Q}_j|^2 (b_{\min}-b)^2,
\end{align}\end{small}where the first inequality is by the upper bound on $g_j$, and the second inequality is because $b_{\min}-b\leq 0.$       

\noindent{\bf Case II:} If $g_j\ge (2|\mathcal{Q}_j|-1)(b_{\max}-b)+\frac{1}{2b_{\max}}$, we have 
\vspace{-0.5cm}
\begin{small}
\begin{align}\nonumber
&|\mathcal{Q}_j|^2l^2-2|\mathcal{Q}_j|lg_j=(|\mathcal{Q}_j|(b_{\max}-b))^2-2|\mathcal{Q}_j|(b_{\max}-b)g_j\cr 
&\qquad\leq (-3|\mathcal{Q}_j|^2+2|\mathcal{Q}_j|)(b_{\max}-b)^2-\frac{|\mathcal{Q}_j|(b_{\max}-b)}{b_{\max}}\cr 
&\qquad\leq (-3|\mathcal{Q}_j|^2+2|\mathcal{Q}_j|)(b_{\max}-b)^2\leq -|\mathcal{Q}_j|^2 (b_{\max}-b)^2,
\end{align}\end{small}where the first inequality is by the lower bound on $g_j$. 

\noindent{\bf Case III:} If $g_j$ does not belong to Cases I and II, then $l=\frac{2g_j-\Psi+\sqrt{\Psi^2+4|\mathcal{Q}_j|-2}}{4|\mathcal{Q}_j|-2}$, which is the unique root of the derivative of \eqref{eq:equilibrium-symmetric-restricted_cost}, and hence it satisfies $l=\frac{1}{2|\mathcal{Q}_j|-1}(g_j+\frac{1}{2(b+l)})$. Therefore, we can write
\vspace{-0.5cm}  
\begin{small}
\begin{align}\nonumber
&|\mathcal{Q}_j|^2l^2-2|\mathcal{Q}_j|lg_j\cr 
&=(\frac{|\mathcal{Q}_j|}{2|\mathcal{Q}_j|-1})^2(g_j+\frac{1}{2(b+l)})^2-\frac{2|\mathcal{Q}_j|}{2|\mathcal{Q}_j|-1}g_j(g_j+\frac{1}{2(b+l)})\cr 
&=(\frac{|\mathcal{Q}_j|}{2|\mathcal{Q}_j|-1})^2\Big[\!-g_j^2(3\!-\!\frac{2}{|\mathcal{Q}_j|})\!-\!\frac{g_j}{b+l}(1\!-\!\frac{1}{|\mathcal{Q}_j|})\!+\!\frac{1}{4(b+l)^2}\Big].
\end{align}\end{small}Now, one can easily see that, if $|g_j|>\frac{1}{b_{\min}}$, then either $g_j\ge \frac{1}{b+l}$ or $g_j\leq \frac{-1}{b+l}$, and the quadratic expression inside of the above brackets is always less than $-g_j^2$. Thus,
\begin{align}\nonumber
(g_j^{\rm NE})^2-g_j^2&=|\mathcal{Q}_j|^2l^2-2|\mathcal{Q}_j|lg_j\cr 
&\leq -(\frac{|\mathcal{Q}_j|}{2|\mathcal{Q}_j|-1})^2g^2_j\leq -(\frac{g_j}{2})^2,      
\end{align}
On the other hand, if $|g_j|\leq \frac{1}{b_{\min}}$, then the quadratic expression inside of the above brackets can be at most $(\frac{2|\mathcal{Q}_j|-1}{|\mathcal{Q}_j|})^2\frac{1}{4b^2_{\min}(3-\frac{2}{|\mathcal{Q}_j|})}$, which implies that
\begin{align}\nonumber
(g_j^{\rm NE})^2\!-\!g_j^2&=|\mathcal{Q}_j|^2l^2\!-\!2|\mathcal{Q}_j|lg_j\cr 
&\leq \frac{1}{4b^2_{\min}(3-\frac{2}{|\mathcal{Q}_j|})}\leq \frac{1}{4b^2_{\min}}.      
\end{align}
Thus $(g_j^{\rm NE})^2\leq \frac{5}{4b^2_{\min}}$, which means that station $j$ remains to be a good station in the NE.

Finally, from all the above cases we have $(g_j^{\rm NE})^2-g_j^2\leq -\mu^2_j, \forall j\in \mathcal{B}$, with the $\mu_j$ as given in the statement of the theorem. Summing this inequality over all the bad stations $j\in \mathcal{B}$ completes the proof. \end{proof} 

As a result of Theorem \ref{thm:load-reduction}, the players charging/discharging strategies at a NE always improves the load balance in the grid. In fact, in practice, such load balancing can be very effective as shown through simulations in Section \ref{sec:simulation}. 

\begin{remark}
\normalfont In general, assigning the EVs optimally to balance the load in a centralized manner is computationally very expensive as it requires solving a mixed nonlinear integer program with the objective function $\sum_{i=1}^{n}C_i(\boldsymbol{a})$ to find the optimal paths, charging stations, and the charge/discharge energy units. However, Theorems \ref{thm:poa} and \ref{thm:load-reduction} suggest that for a large number of EVs the optimal assignment can be approximated within a constant factor by a solution where each EV selfishly minimizes its \textit{own} cost. This can be done much more efficiently as now each EV minimizes its cost over only its \emph{own} actions.
\end{remark}

\section{Stochastic Ground Load with Prospect EVs}\label{sec:stochatic-ground}

In this section, we consider the EVs' interaction game under a more realistic grid scenario with an uncertain ground load environment and study the effect of EVs' behavioral decisions on the overall performance of the smart grid. Toward this goal, we assume that the induced ground load at each station $g_j$, $j\in \mathcal{M}$, which is due to industrial, residential, or commercial users is a random variable with some unknown distribution $G_j$. Indeed, in a smart grid, a good portion of the energy generated and injected into the grid will stem from renewable resources such as wind turbines or solar panels. Since the amount of such renewable energy highly depends on the environment, such as weather conditions, which is a stochastic phenomenon, the induced renewable energy also changes stochastically at various locations \cite{etesami2016stochastic}. On the other hand, the energy consumption of residential or industrial users normally follows certain stochastic patterns during specific time slots of a day (e.g., more consumption during early evening hours and less after midnight). Since the ground loads at different stations are mainly influenced by the grid components within their vicinity, for sufficiently distant stations, we may assume that the induced ground loads are stochastically independent. Under this independency assumption, we study the optimality of the EVs' game under stochastic ground load.

It is worth noting that, in general, the PoA of the EVs game is a function of its underlying parameters such as ground loads or the number of EVs. Therefore, in the presence of stochastic ground loads, the PoA will also be a random variable. As it was proposed in \cite{kempton2005vehicle}, the grid authority can use EVs to balance the load on the grid by charging when demand is low and selling power back to the grid (discharging) when demand is high. To this end, the following theorem provides an estimate on the required number of EVs to be added into the network in order to keep the grid social cost within a constant factor of its optimal value (i.e., a low PoA) with high probability.
\begin{theorem}\label{thm:chernoff}
Let $G_j, j=1,\ldots,m$ be stochastically independent ground loads with support in $[-K,K]$ such that $\mathbb{E}[G_j]=\mu_j$, and $Var[G_j]=\sigma^2_j$. Then, participating $n\ge 4.5[\sum_{j=1}^{m}(\mu_j^2+\sigma_j^2)\!+\!K\sqrt{m\ln(\frac{1}{\epsilon})}]$ EVs in the grid assures that $\emph{PoA}\leq 4+12b_{\max}^2$, with high probability $1-\epsilon$.    
\end{theorem}
\begin{proof}
The proof can be found in Appendix \ref{appx:chernoff}.
\end{proof}

\subsection{Prospect-Theoretic Analysis of the EVs' Game}

In this section, we take into account the subjective behavior of EV owners under uncertain energy availability. In this regard, there is a strong evidence \cite{kahneman1979prospect} that, in the real-world, human decision-makers do not make decisions based on expected values of outcomes evaluated by actual probabilities, but rather based on their perception on the potential value of losses and gains associated with an outcome. Indeed, using PT, the authors in \cite{kahneman1979prospect} showed that human individuals such as EV owners, will often overestimate low probability outcomes and underestimate high probability outcomes. This phenomenon, known as \textit{weighting} effect in PT, reflects the fact that EV owners usually have subjective views on uncertain outcomes such as energy availability at the charging stations. Moreover, there is evidence that in reality humans perceive and frame their losses or gains with respect to a reference point using their own, individual and subjective value function. As an example, risk-averse EV owners consider any energy price higher than that when the grid operated in its balanced condition as a loss and overestimate it. This is a consequence of the so-called \emph{loss aversion} behavior which leads different EVs to select different reference points and evaluate their gains/losses according to them. Such reference-dependent loss aversion behavior can be explained under the \textit{framing} effect in PT which differs from CGT that assumes players are rational agents who aim to minimize their \emph{expected} losses.

In fact, PT has been successfully applied in many problems with applications both in engineering and economics. For instance, the authors in \cite{hota2016fragility} study humans' behavioral decisions in the presence of failure risk in a common-pool resource game. It has been shown in \cite{etesami2016stochastic} that taking into account the subjective behavior of prosumers (joint prosumer-consumer) in the smart grid can substantially change the energy management and distribution pattern compared to the conventional expected utility methods. We refer to  \cite{barberis2013thirty} and \cite{wakker2010prospect} for a comprehensive survey and recent results on PT in economics and other fields. Therefore, to capture such behavioral decisions, we use the following definition from PT \cite{kahneman1979prospect}:
\begin{definition}\label{def:prospect}
\normalfont Any EV $i$ has a reference point $z_r^i$ and two corresponding functions $w_i:[0,1]\to \mathbb{R}$ and $v_i(z,z_r^i):\mathbb{R}^2\to \mathbb{R}$, known as \emph{weighting} and \emph{valuation} functions. The \emph{expected prospect} of a random variable $Z$ with outcomes $z_1,z_2,\ldots,z_k$, and corresponding probabilities $p_1,p_2,\ldots,p_k$, for electric vehicle $i$ is given by 
\begin{align}\nonumber
\mathbb{E}^{^{\rm PT}}[Z]:=\sum_{\ell=1}^{k}w_i(p_{\ell})v_i(z_{\ell},z_r^i).
\end{align}
\end{definition}
In general, the value function that passes through the reference point is $S$-shaped and asymmetrical. This means that the value function is steeper for losses than gains indicating that losses outweight gains. Two of the widely used weighting and valuation functions in the PT literature are known as \emph{Prelec} weighting function and \emph{Tversky} valuation function defined by \cite{prelec1998probability,al2008note},
\begin{align}\label{eq:weight-value}
&v_i(z,z^i_r)=\begin{cases} (z-z^i_r)^{c_1} & \mbox{if} \ \ z\ge z^i_r, \\ 
-c_2(z^i_r-z)^{c_3} & \mbox{if} \ \ z<z^i_r,\end{cases}\cr 
&w_i(p)=\exp(-(-\ln p)^{c}),
\end{align}where $0<c\leq 1$ is a constant denoting the distortion between subjective and objective probability. Here, $c_1,c_3\in (0,1)$ determine the curvature of the value function in gains and losses, respectively, and capture humans behaviour as risk averse in gains and risk seeking in losses justified by behavioral economics \cite{booij2010parametric,al2008note,wakker2010prospect}. On the other hand, loss aversion is typically captured by the parameter $c_2 > 1$ which reflects the fact that human usually perceive losses much more than gains and outweight them.\footnote{The behavior under $c_2\in (0,1)$ is often referred to as gain seeking \cite{wakker2010prospect}.} Moreover, we assume that the reference energy price for EV $i$ is given by 
\begin{align}\label{eq:reference}
z_r^i:=f(\sum_{j\in\mathcal{Q}_{q_i}}l_j)-f(\!\!\!\sum_{j\in\mathcal{Q}_{q_i}\setminus\{i\}}\!\!\!l_j),
\end{align}
which is the price that EV $i$ expects to pay in station $q_i$ given that this station operates in its complete balanced condition (i.e., $g_{q_i}=0$). In particular, anything above or below this reference price is considered as loss or gain for that EV and is measured by value function $v(z,z^i_r)$.

To formulate the EVs' interaction game using PT, we assume that the ground load at station $j\in\mathcal{M}$ follows a discrete distribution $G_j$ with zero mean and a probability mass function $h_j(\cdot)$. Let $Z$ be the random variable $Z:=f(\sum_{j\in\mathcal{Q}_{q_i}}\!\!l_j-G_{q_i})-f(\sum_{j\in\mathcal{Q}_{q_i}\setminus\{i\}}\!l_j-G_{q_i})$, and $z(\theta)$ be the realization of $Z$ when the ground load at station $q_i$ is $\theta$. Therefore, by Definition \ref{def:prospect}, the perceived \emph{prospect} gain/loss by EV $i$ equals 
\begin{align}\label{eq:E[Z]}
\mathbb{E}^{^{\rm PT}}[Z]=\sum_{\theta}w_i(h_{q_i}(\theta))v_i(z(\theta),z_r^i).
\end{align}
On the other hand, as it has been shown in \cite{kHoszegi2006model}, that gains and losses are not all that EVs care about. In other words, not only the sensation of gain or avoided loss does affect the payoff function for an EV $i$ but so does the actual energy price that EV $i$ pays to satisfy its need. Therefore, in contrast to the prior formulation based on a value function defined solely over gains and losses, we take preferences also into the cost functions by assuming that the overall cost to EV $i$ with reference point $z^i_r$ is given by
\begin{align}\label{eq:EV-PT-cost}
C^{\rm PT}_i(\boldsymbol{a}_i,\boldsymbol{a}_{-i})&\!=\!\sum_{e\in P_i}\!c_{e}(n_e)\!+\!\frac{|\mathcal{Q}_{q_i}|}{\sigma_{q_i}}\!+\!\ln\Big(\frac{b_{\max}}{b_i+l_i}\Big)\cr 
&\qquad+z^i_r+\mathbb{E}^{^{\rm PT}}[Z],
\end{align}where $z^i_r$ and $\mathbb{E}^{^{\rm PT}}[Z]$ are given by \eqref{eq:reference} and \eqref{eq:E[Z]}, respectively. Here, each EV $i$ aims to minimize its own prospect cost given by \eqref{eq:EV-PT-cost} by choosing an appropriate action $\boldsymbol{a}_i$. The following theorem shows that despite extra nonlinearity of the weighting and reference effects in the players' cost functions, the EVs' game under PT still admits a pure NE.  
\begin{theorem}\label{thm:existence-pure_PT}
For the quadratic pricing $f(x)=x^2$, the EVs' game under PT admits a pure-strategy NE. In particular, the best response dynamics converge to one of such NE points.
\end{theorem}
\begin{proof}
For the \emph{quadratic} pricing function we have 
\begin{align}\nonumber
z(\theta)-z^i_r=(\sum_{j\in\mathcal{Q}_{q_i}}\!l_j-\theta)^2-(\!\!\!\sum_{j\in\mathcal{Q}_{q_i}\setminus\{i\}}\!\!\!\!\!l_j-\theta)^2-z_r^i=l_i\theta.
\end{align}Substituting this relation into \eqref{eq:EV-PT-cost}, we obtain 
\begin{align}\nonumber
C^{\rm PT}_i(\boldsymbol{a}_i,\boldsymbol{a}_{-i})&=\sum_{e\in P_i}c_{e}(n_e)+\frac{|\mathcal{Q}_{q_i}|}{\sigma_{q_i}}+\ln\Big(\frac{b_{\max}}{b_i+l_i}\Big)\cr 
&\qquad+z_r^i+\underbrace{\sum_{\theta}w_i(h_{q_i}(\theta))\hat{v}_i(l_i\theta)}_\text{$\tilde{C}^{\rm PT}(\boldsymbol{a})$},
\end{align}
where $\hat{v}_i(x)=x^{c_1}$ if $x\ge 0$, and $\hat{v}_i(x)=-c_2|x|^{c_3}$, otherwise. Now consider the function $\Psi(\cdot)$ defined by
\begin{align}\nonumber
\Psi(\boldsymbol{a}_i,\boldsymbol{a}_{-i})&\!=\!\sum_{e\in \mathcal{E}}\sum_{x=1}^{n_e}c_{e}(x)\!+\!\sum_{\ell=1}^{m}\frac{|\mathcal{Q}_{\ell}|(|\mathcal{Q}_{\ell}|\!+\!1)}{2\sigma_{\ell}}\cr
&\qquad+\sum_{i=1}^{n}\ln\left(\!\frac{b_{\max}}{b_i+l_i}\!\right)+\sum_{\ell=1}^{m}(\sum_{j\in\mathcal{Q}_{\ell}}l_j)^2\cr 
&\qquad+\underbrace{\sum_{i=1}^{n}\sum_{\theta}w_i(h_{q_i}(\theta))\hat{v}_i(l_i\theta)}_\text{$\tilde{\Psi}(\boldsymbol{a})$}.
\end{align}
We argue that this function is an exact potential function for the EVs' game under PT. In fact, if we did not have the prospect terms $\tilde{C}^{\rm PT}(\boldsymbol{a})$ and $\tilde{\Psi}(\boldsymbol{a})$ in the structure of $C^{\rm PT}_i(\boldsymbol{a})$ and $\Psi(\boldsymbol{a})$, the proof would immediately follow by the same lines of argument as in the proof of Theorem \ref{thm:existence-pure}. However, for the quadratic pricing, since the term $\sum_{\theta}w_i(h_{q_i}(\theta))\hat{v}_i(l_i\theta)$ is a player specific function which \emph{only} depends on action of player $i$ and is uncorrelated from $\boldsymbol{a}_{-i}$, we easily get $\tilde{C}^{\rm PT}(\boldsymbol{a}_i,\boldsymbol{a}_{-i})-\tilde{C}^{\rm PT}(\boldsymbol{a}'_i,\boldsymbol{a}_{-i})=\tilde{\Psi}^{\rm PT}(\boldsymbol{a}_i,\boldsymbol{a}_{-i})-\tilde{\Psi}^{\rm PT}(\boldsymbol{a}'_i,\boldsymbol{a}_{-i})$. This shows that $\Psi(\cdot)$ is indeed an exact potential function for the EVs' game under PT and quadratic pricing. As a result, any minimizer of $\Psi(\cdot)$ is a pure-strategy NE of the EVs' game. In particular, since the action sets of players are compact in their own ambient space, this immediately implies that the sequence of best responses of EVs will converge to a pure-strategy NE \cite{monderer1996potential}.    
\end{proof}
Here, we should mention that if we use different pricing functions or assume other sources of uncertainty such as randomness in players' actions, then the EVs' game under PT will not necessarily admit a pure-strategy NE. One of the challenges of analyzing the proposed EV game under PT is the extra nonlinearities in the players' cost functions which stem from weighting and framing effects. This further complicates the analysis of the PoA under PT. For instance, as opposed to CGT, the PoA of the game with prospect cost functions will now depend on the specific choice of weighting functions and varying reference points. In the next section, we provide some numerical results to study the PoA of the EVs' game under both CGT and PT. 

\section{Simulation Results}\label{sec:simulation}

For our simulations, we choose the traffic network to be as in Figure \ref{fig:Simulation} with $5$ directed roads, and $3$ charging stations. We assume i.i.d Gaussian distributions $G_j\sim N(0,10)$ for the ground load at different stations. Also, for simplicity, we assume that all the EVs are identical with $b_{\max}=5, b_{\min}=0.1$, and $b_i=3, \forall i$, who want to travel from the origin $s$ to the destination $t$.

\subsection{PoA under Classical Game Theory}

In Figure \ref{fig:PoA_CGT}, we illustrate how the PoA under CGT changes as more EVs join the grid and compare the outcomes for different choices of pricing and latency functions. Here, we let the number of EVs increase from $n=2$ to $n=9$, and compute the PoA when the nonlinearity of the pricing function increases from $f(x)=x^{2/3}$ to $f(x)=x^{8/3}$. Moreover, we consider the effect of linear latency function $c_e(x)=5x+10$ and quadratic latency function $c_e(x)=5x^2+10$ on the PoA. As it can be seen, joining more EVs monotonically reduces the PoA as was expected by Theorem \ref{thm:poa} for the case of linear latency and quadratic pricing functions. However, it turns out that the PoA generally increases as the nonlinearities of the pricing and latency functions increase. In particular, the mismatch between the degree of nonlinearity of the pricing and latency functions degrades PoA. Hence, to achieve a high grid performance in terms of social cost, the grid authority should relatively match the energy price with the latency costs.  

Figure \ref{fig:Balance_CGT} illustrates the percentage of load balance improvement (i.e., $100\frac{V_0-V_{\rm NE}}{V_0}\%$) for the worst achieved NE corresponding to each of the cases in Figure \ref{fig:PoA_CGT}. It is interesting to see that there is a tradeoff between the PoA and the load balance. For instance, for the linear latency function $c_e(x)=5x+10$, the pricing function $f(x)=x^{2/3}$ (red dashed line in both figures) achieves the best PoA and the worst load balancing performance. In fact, for the linear latency function, it can be seen that the quadratic pricing $f(x)=x^2$ (dashed black curve) performs very well both in terms of PoA and load balancing. However, it should be noted that the best performance among the above cases is achieved for the pricing rule $f(x)=x^{4/3}$ (solid blue curve).

\begin{figure}[t!]
\begin{center}
\includegraphics[totalheight=.14\textheight,
width=.22\textwidth,viewport=120 0 520 400]{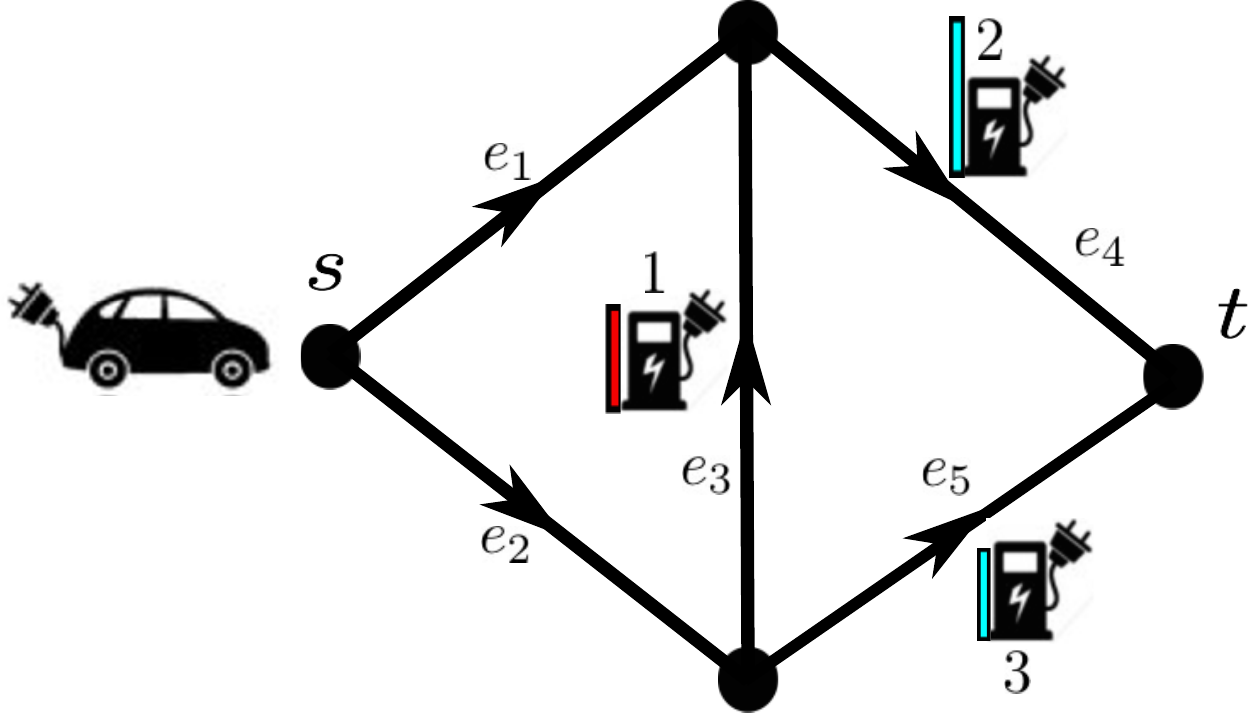} \hspace{0.4in}
\end{center}\vspace{-0.1cm}
\caption{Network structure for the simulations.}
\label{fig:Simulation}
\end{figure}

In Table \ref{table:lin-prog} we have listed the worst NE strategies and social cost, as well as the optimal social cost for $n=9$ vehicles. As an example the NE strategy for player 1 is to take the route $P_2=(e_2,e_5)$, join station $3$, and charge its battery by $l_1=0.46$ energy units. In this table, the initial realized random ground loads at stations $Q_1$, $Q_2$, and $Q_3$ are $g_1=0.937$, $g_2=-11.223$, and $g_3=3.061$, respectively. Therefore, the initial load imbalance equals $V_0=136.207$, while the ground loads at the NE at these stations are given by $g^{\rm NE}_1=0.123$, $g^{\rm NE}_2=-5.007$, and $g^{\rm NE}_3=-1.229$. As a result, the load imbalance at this NE equals $V_{\rm NE}=26.60$ which is substantially lower than the initial load imbalance $V_0=136.207$ ($84\%$ improvement).

\begin{figure}[t!]
\vspace{-1.4cm}
\begin{center}
\includegraphics[totalheight=.27\textheight,
width=.42\textwidth,viewport=-10 0 390 400]{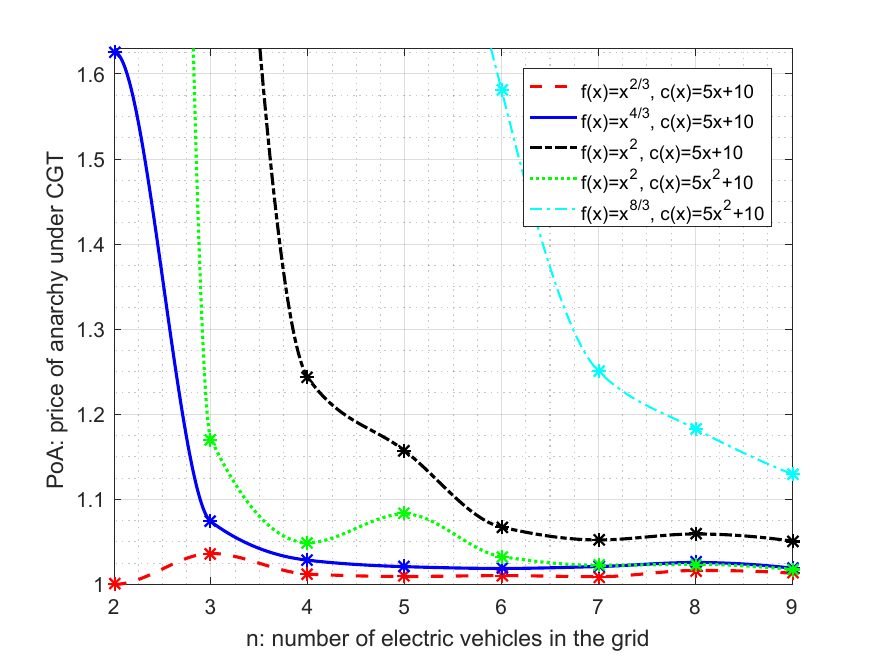} \hspace{0.4in}
\end{center}\vspace{-0.1cm}
\caption{PoA under CGT for different number of EVs, pricing functions, and latency functions.}
\label{fig:PoA_CGT}
\end{figure}

\begin{figure}[t!]
\vspace{-1.1cm}
\begin{center}
\includegraphics[totalheight=.27\textheight,
width=.38\textwidth,viewport=0 0 400 400]{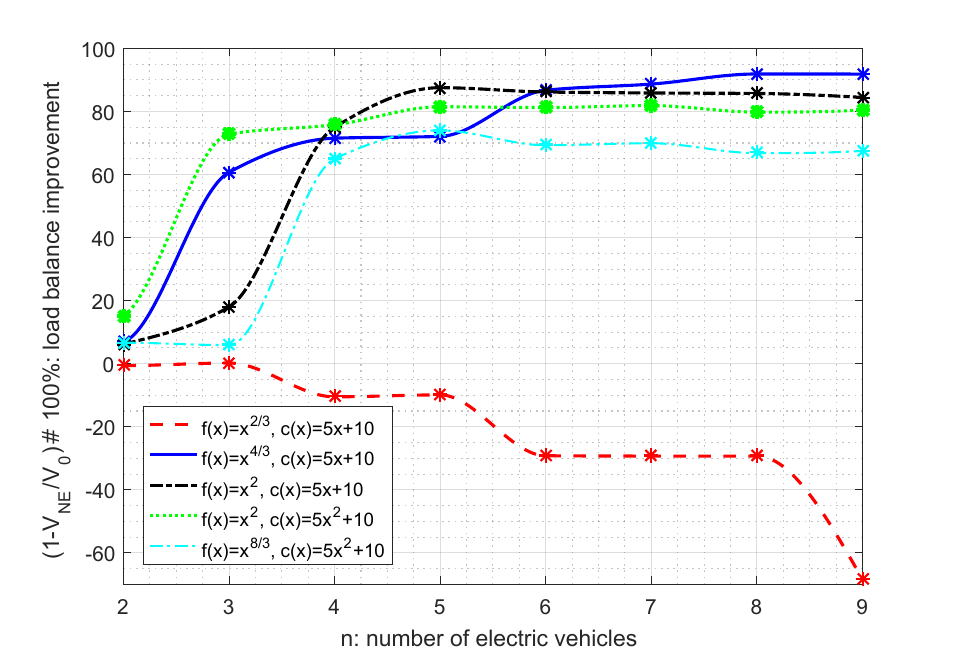} \hspace{0.4in}
\end{center}\vspace{-0.1cm}
\caption{Load balance improvement under CGT and different pricing/latency functions.}
\label{fig:Balance_CGT}
\end{figure}

\begin{table}[t!]
\tiny
\vspace{0.5cm}
\caption{Pure NE for $n=9$ EVs, three paths $P_1=(e_1,e_4)$, $P_2=(e_2,e_5)$, $P_3=(e_2,e_3,e_4)$, and three stations $Q_1$, $Q_2$, and $Q_3$.} \vspace{0.2cm}
\hspace{1cm}\begin{tabular}{c c c c c c c c} 
\hline\hline 
$n$ & $P_1$ & $P_2$ & $P_3$ & $Q_1$ & $Q_2$ & $Q_3$& $l_i$\\ [0.5ex] 
\hline 
$1$ & 0 & 1 & 0 & 0 & 0 & 1 & 0.46 \\ 
$2$ & 0 & 1 & 0 & 0 & 0 & 1 & 0.46 \\
$3$ & 0 & 1 & 0 & 0 & 0 & 1 & 0.46 \\
$4$ & 0 & 1 & 0 & 0 & 0 & 1 & 0.46 \\
$5$ & 0 & 0 & 1 & 1 & 0 & 0 & 1.06 \\
$6$ & 1 & 0 & 0 & 0 & 1 & 0 & -1.55 \\
$7$ & 1 & 0 & 0 & 0 & 1 & 0 & -1.55 \\
$8$ & 1 & 0 & 0 & 0 & 1 & 0 & -1.55 \\
$9$ & 1 & 0 & 0 & 0 & 1 & 0 & -1.55 \\
\hline
NE & 863.53  & & & & & OPT & 821.77\\
\hline 
\end{tabular}\label{table:lin-prog}
\end{table}\vspace{1cm}

\subsection{PoA under Prospect Theory}

Here, we evaluate the effect of PT on the PoA and load balancing. We set $c_e(x)=5x+10$ and $f(x)=x^2$, and consider $n=6$ EVs over the network of figure \ref{fig:Simulation}. Moreover, we assume that all the EVs have the same weighting and valuation functions given by \eqref{eq:weight-value} with parameters $(c,c_1,c_2,c_3)$. Figure \ref{fig:c_PoA} illustrates the effect of probability distortion parameter $c$ for fixed values of $(c_1,c_2,c_3)=(0.88,2.25,0.88)$ which are estimated based on experimental studies on human subjects \cite{booij2010parametric}. As can be seen from the top figure, the PoA has a complicated nonlinear relation with the distortion probability parameter. One possible reason is that for mid-ranges of the probability distortion parameter, the system has a large number of NEs which results in worse performance in terms of PoA. However, the induced load in the grid stations monotonically decreases as EVs become more rational (i.e., $c$ approaches to 1). In particular, for small values of $c$, the EVs start to charge or discharge more aggressively which will fully imbalance the station loads. This is because, for very low ranges of $c$, the EVs behave fully irrational and start to make a profit by completely ignoring their travel costs and joining the profitable stations to buy/sell energy at a very low/high price.  
 
\begin{figure}[t!]
\vspace{-1.2cm}
\begin{center}
\includegraphics[totalheight=.29\textheight,
width=.43\textwidth,viewport=0 0 500 500]{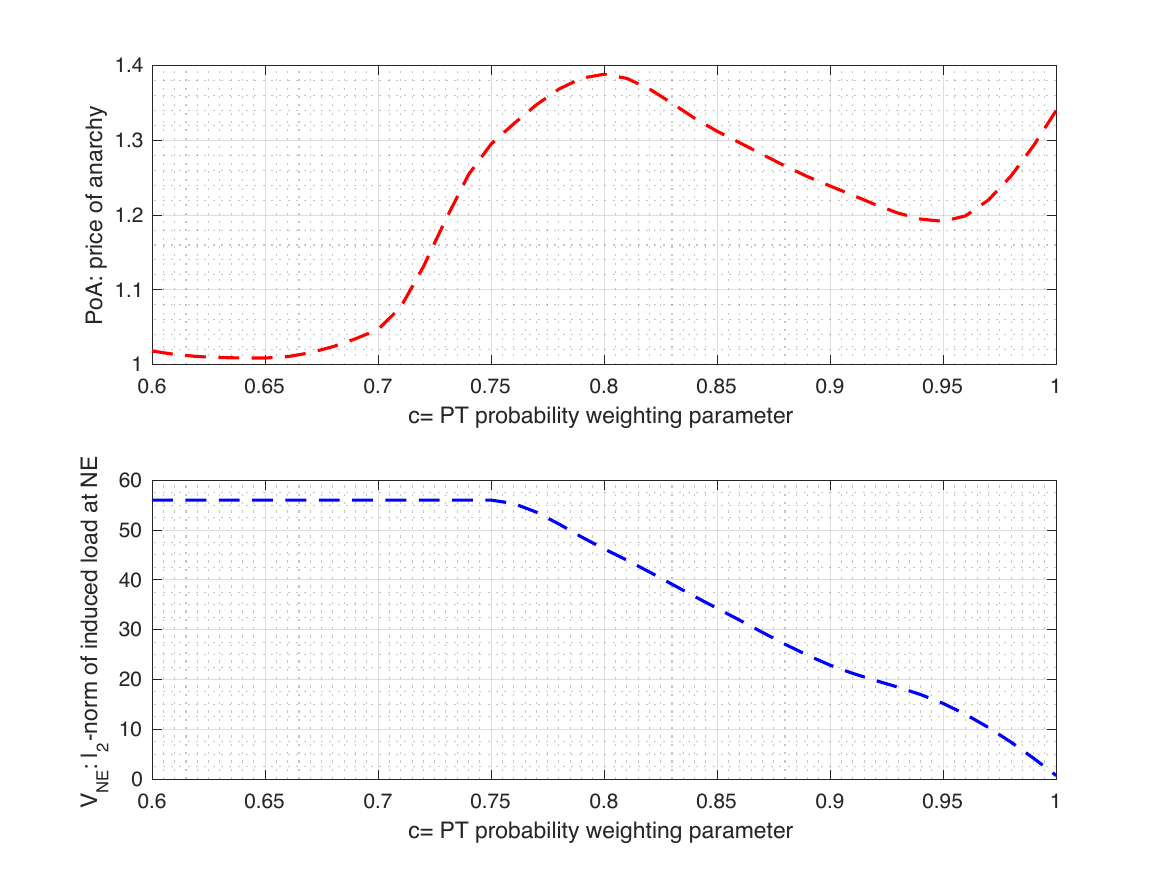} \hspace{0.4in}
\end{center}\vspace{-0.1cm}
\caption{The red curve illustrates the PoA for different values of the PT probability weighting parameter $c$. The blue curve depicts the $l_2$-norm of the induced load at the worst NE (i.e., $V_{\rm NE}=\sum_{j}(g_j^{\rm NE})^{2}$) versus the PT weighting parameter $c$.}
\label{fig:c_PoA}
\end{figure}

Finally, in Figure \ref{fig:PT_different_parameters} we have illustrated the effect of different PT parameters $(c,c_1,c_2,c_3)$ estimated from experimental studies \cite{booij2010parametric} on the PoA and total NE induced load in the stations.\footnote{In Figure \ref{fig:PT_different_parameters}, we have scaled down the total NE induced load $V_{\rm NE}$ by a factor of $0.01$.} Here, we set $A=(0.75,0.68,2.54,0.74)$, $B=(0.75,0.81,1.07,0.8)$, $C=(0.75,0.71,1.38,0.72)$, $D=(0.75,0.86,1.61,1.06)$, and $E=(0.75,0.88,2.25,0.88)$. In particular, the last bar corresponds to the selection of $(c,c_1,c_2,c_3)=(1,1,1,1)$ which is for the case of risk neutral EVs. As it can be seen, for the above set of parameters, the grid benefits the most (both in terms of PoA and load balance) when the EVs are risk neutral (as it was the underlying assumption in modeling the EVs interaction game). The worst-case situations occur for the EV owners whose subjective valuation lie in group parameters $D$ and $E$. This suggest that for such type of EVs, one must modify the pricing rules in order to take into account the negative effects of EVs behavioral decisions.

\begin{figure}[t!]
\vspace{-1cm}
\begin{center}
\includegraphics[totalheight=.25\textheight,
width=.38\textwidth,viewport=0 0 400 400]{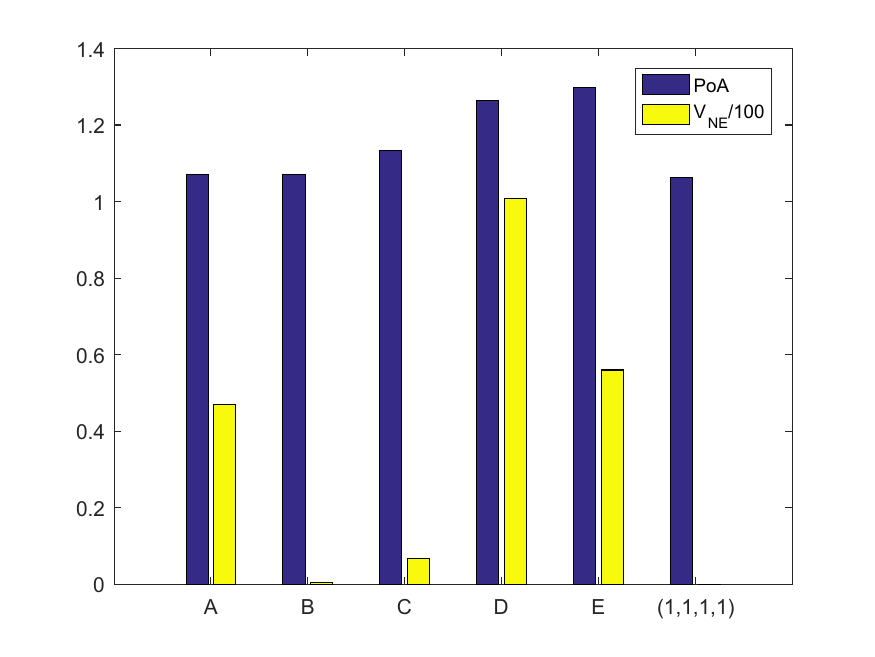} \hspace{0.4in}
\end{center}\vspace{-0.1cm}
\caption{PoA (blue bar) and total induced load at the worst NE (yellow bar) for different set of PT parameters estimated from behavioral studies.}
\label{fig:PT_different_parameters}
\end{figure}

\section{Conclusions}\label{sec:conclusion}

In this paper, we have studied the interaction of selfish electric vehicles in smart grids. We have formulated a noncooperative game between the EVs and, then, we have shown that the game admits a pure-strategy NE. Then, we have shown that the PoA of the game is bounded above by the ``variance" of the ground load divided by the total number of vehicles. This, in turn, implied that for a large number of EVs in the grid, the entire system operates very close to its optimal condition with the minimum social cost, despite the fact that EVs are selfish identities. In particular, we have obtained a tighter upper bound for the PoS of the EVs' game and showed that any achieved NE balances the load further across the grid. We have extended our results to the case where the ground load is stochastic and incorporated the subjective behavior of EVs using PT into our model. Simulation results showed that, under realistic grid scenarios with subjective EVs, quadratic pricing is more suitable for a large number of EVs, while for fewer EVs, exponential pricing would be a better choice. 

\noindent
{\bf Discussion and Future Directions:} In this work, we have considered the EVs' game with a basic cost structure which can potentially be extended to more general settings. For instance, we have only considered the case where an EV can join at most one station during its trip. This can be relaxed further to the case where each EV can join up to at most $k$ stations, for some integer constant $k$. For this purpose, one can extend the action space of EV $i$ from $(P_i,q_{i},l_i)$ to $(P_i,\{q_{i_r}\}_{r=1}^{k},\{l_{i_r}\}_{r=1}^{k})$, and augment its cost function \eqref{eq:EV-cost-Queue} by adding the terms $\sum_{r=1}^{k}\frac{|\mathcal{Q}_{q_{i_r}}|}{\sigma_{q_{i_r}}}$
and $\sum_{r=1}^{k}[f_{q_{i_r}}\!(\sum_{j\in\mathcal{Q}_{q_{i_r}}}\!\!\!l_{j_r}-g_{q_{i_r}}\!)\!-\!f_{q_{i_r}}(\sum_{j\in\mathcal{Q}_{q_{i_r}}\!\!\setminus\!\{i\}}l_{j_r}-g_{q_{i_r}})]$. In this case, one can again use the same potential function as in \eqref{eq:potential} to show the existence of pure Nash equilibrium and obtain similar PoA bounds (with possibly an additional factor in terms of $k$). However, in this extended formulation, each EV must solve a more complex optimization problem to find its optimal action. 

This work opens several directions for future research. To name a few, 

\begin{itemize}
\item One can consider implementing the EVs' game under incomplete information. Perhaps, one way of modeling incomplete information into our model is to use the framework of Bayesian games with EVs grouped into different types based on their origin-destination, where each EV aims to minimize its own expected cost having access to only a prior distribution of other EVs' types. In particular, we did not explore the effect of information structure on the game outcome and simply assumed that the data platform provides all the information truthfully to the EVs. However, it is interesting to see how the platform can leverage the information structure to bring the system to a more efficient and ``load-balanced" equilibrium.
\item One can consider an online optimization version of our setting in which the number of EVs is not fixed, and the vehicles sequentially join the system by best responding to the current state of the system. In that regard, the results of \cite{harks2009competitive} seem to be a good starting point.
\item One can extend our model to the case in which there is a capacity constraint on the stations. While this can be remedied up to some extent by assuming that the rate of process of each station is proportional to its capacity, however, full analysis of the EV game under hard capacity constraint is an interesting direction of research.
\item Finally, studying the EVs' behavioral decisions in the presence of mixed-strategies is very interesting. In such scenarios uncertainty will stem, not only from ground loads but also from EVs' probabilistic decisions. Therefore, one would expect to observe more deviations between PT and CGT as it has been shown in \cite{etesami2016stochastic} for a different grid setting.
\end{itemize}


\bibliographystyle{unsrt}
\bibliography{autosam}

\appendix
\section{Appendix}\label{ap-static-preliminary-proofs}

\begin{lemma}\label{apx:thm:poa}
For $OPT_3, NE_3$ given in Theorem \ref{thm:poa}, we have
\begin{align}\nonumber
&\sum_{i=1}^{n}\Big[(\!\!\!\sum_{k\in \mathcal{Q}_{q^*_{i}}\!\setminus\!\{i\}}\!\!\!\!\!\!\!l_k\!+\!l^*_i\!-\!g_{q^*_{i}})^2\!-\!(\!\!\!\!\!\!\sum_{k\in \mathcal{Q}_{q^*_{i}}\!\setminus\!\{i\}}\!\!\!\!\!\!l_k\!-\!g_{q^*_{i}})^2\!+\!\ln\big(\frac{\bar{b}_i}{b_i+l^*_i}\big)\Big]\cr 
&\qquad\leq \sqrt{\left(\gamma+NE_3\right)\left(\gamma+OPT_3\right)}+OPT_3\cr 
&\qquad+\delta(\sqrt{\gamma+NE_3}+\sqrt{\gamma+OPT_3})+(\delta^2+4nb_{\max}^2),
\end{align}
where $b_{\max}=\max_i\bar{b}_i$, $\delta^2\!:=\frac{\sum_{j}g^2_j}{2}$, and $\gamma\!:=\delta^2+nb^2_{\max}$.
\end{lemma}
\begin{proof}
let $L^*_j:=\sum_{k\in \mathcal{Q}^*_j}l^*_k$ be the aggregate load induced by the optimal solution at station $j$. We have
\begin{align}\nonumber
OPT_3&-\sum_{i=1}^{n}\ln\Big(\frac{\bar{b}_i}{b_i+l^*_i}\Big)\cr 
&=\sum_{i=1}^{n}\!\Big[(\!\!\!\sum_{k\in \mathcal{Q}^*_{q^*_{i}}}\!l^*_k-g_{q^*_{i}})^2\!-\!(\!\!\!\!\sum_{k\in \mathcal{Q}^*_{q^*_{i}}\setminus\{i\}}\!\!\!\!\!\!l^*_k-g_{q^*_{i}})^2\Big]\cr
&=\sum_{i=1}^{n}\Big((l^*_i)^2+2l^*_i\big(\!\!\!\!\sum_{k\in \mathcal{Q}^*_{q^*_{i}}\setminus\{i\}}\!\!\!l^*_k-g_{q^*_i}\big)\Big)\cr 
&=\sum_{j=1}^{m}\Big(2\big(\!\!\sum_{k\in \mathcal{Q}^*_j}l^*_k\big)^2\!-\!\!\sum_{k\in \mathcal{Q}^*_j}(l^*_k)^2\Big)\!-\!2\sum_{i=1}^{n}g_{q^*_i}l_i^*\cr
&=\sum_{j=1}^{m}\Big(2\big(L^*_j)^2-\sum_{k\in \mathcal{Q}^*_j}(l^*_k)^2\Big)-2\sum_{j=1}^{m}g_jL^*_j\cr 
&\geq 2\sum_{j=1}^{m}\!\big(L^*_j)^2\!-\!\sum_{i=1}^{n}(l^*_i)^2\!-\!2\sqrt{\sum_{j=1}^{m}g^2_j\sum_{j=1}^{m}(L^*_j)^2}, 
\end{align}
where the last inequality is by Cauchy–Schwarz inequality. Now since for every $i$ we have $l^*_i\leq b_{\max}$ and $\ln\big(\frac{\bar{b}_i}{b_i+l^*_i}\big)\ge 0$, using the above relation we obtain, 
\begin{align}\nonumber
OPT_3&\ge 2\sum_{j=1}^{m}\big(L^*_j)^2-2\sqrt{\sum_{j=1}^{m}g^2_j\sum_{j=1}^{m}(L^*_j)^2}-nb^2_{\max}.
\end{align}
Defining $A:=\sqrt{\sum_{j=1}^{m}(L_j^*)^2}$, we can rewrite the above inequality as $A^2-\sqrt{\sum_{j=1}^{m}g^2_j}A-\frac{1}{2}(nb_{\max}^2+OPT_3)\leq 0$. As this quadratic polynomial is nonpositive, its discriminant must be nonnegative. Thus
\begin{align}\nonumber
\Delta^*:=\sum_{j=1}^{m}g^2_j+2(nb_{\max}^2+OPT_3)=2\gamma+2OPT_3\ge 0,
\end{align}
In particular, solving this quadratic inequality for $A$, we obtain $A\leq \frac{1}{2}(\sqrt{\sum_{j=1}^{m}g^2_j}+\sqrt{\Delta^*})$, or equivalently 
\begin{align}\label{eq:OPT_delta}
\sqrt{\sum_{j=1}^{m}(L_j^*)^2}\leq \frac{1}{2}(\sqrt{\sum_{j=1}^{m}g^2_j}+\sqrt{\Delta^*}). 
\end{align}
Using identical steps for the Nash equilibrium we obtain,
\begin{align}\label{eq:NE_delta}
\sqrt{\sum_{j=1}^{m}L_j^2}\leq \frac{1}{2}(\sqrt{\sum_{j=1}^{m}g^2_j}+\sqrt{\Delta}), 
\end{align}
where $L_j:=\sum_{k\in \mathcal{Q}_j}l_k$ is the aggregate load induced by the NE at station $j$, and $\Delta:=2\gamma+2NE_3$. Now we have
\begin{align}\nonumber
&\sum_{i=1}^{n}\!\Big[(\!\!\!\sum_{k\in \mathcal{Q}_{q^*_{i}}\!\setminus\!\{i\}}\!\!\!\!\!\!\!l_k\!+\!l^*_i\!-\!g_{q^*_{i}})^2\!-\!(\!\!\!\!\!\!\sum_{k\in \mathcal{Q}_{q^*_{i}}\!\setminus\!\{i\}}\!\!\!\!\!\!\!l_k\!-\!g_{q^*_{i}})^2\!+\!\sum_{i=1}^{n}\ln\Big(\frac{\bar{b}_i}{b_i\!+\!l^*_i}\Big)\Big]\cr  
&=2\sum_{i=1}^{n}l_i^*\!\!\!\sum_{k\in \mathcal{Q}_{q^*_{i}}\!\setminus\!\{i\}}\!\!\!\!\!\!l_k\!+\!\sum_{i=1}^{n}(l^*_i)^2\!-\!2\sum_{i=1}^{n}g_{q^*_i}l_i^*\!+\!\sum_{i=1}^{n}\ln \Big(\frac{\bar{b}_i}{b_i\!+\!l^*_i}\Big)\cr 
&=2\sum_{i=1}^{n}l_i^*\!\!\!\sum_{k\in \mathcal{Q}_{q^*_{i}}\!\setminus\!\{i\}}\!\!\!\!\!\!l_k-2\sum_{i=1}^{n}l_i^*\!\!\sum_{k\in \mathcal{Q}^*_{q^*_{i}}\!\setminus\!\{i\}}\!\!\!\!l^*_k+OPT_3\cr
&=2\sum_{i=1}^{n}l_i^*\Big(L_{q^*_i}-l_i\boldsymbol{1}_{\{q_i=q^*_i\}}-(L^*_{q^*_i}-l^*_{i})\Big) +OPT_3\cr 
&=2\sum_{i=1}^{n}l_i^*(l^*_{i}-l_i\boldsymbol{1}_{\{q_i=q^*_i\}})+2\sum_{j=1}^{m}L_j^*(L_{j}-L^*_{j}) +OPT_3\cr 
&\leq 4nb_{\max}^2+2\sum_{j=1}^{m}L^*_jL_j+OPT_3\cr
&\leq 4nb_{\max}^2\!+\!2\sqrt{\sum_{j=1}^{m}(L^*_j)^2\sum_{j=1}^{m}L_j^2}+OPT_3\cr 
&\leq 4nb_{\max}^2\!+\!\frac{1}{2}(\sqrt{\sum_{j=1}^{m}g^2_j}\!+\!\sqrt{\Delta})(\sqrt{\sum_{j=1}^{m}g^2_j}\!+\!\sqrt{\Delta^*})+OPT_3\cr 
&=OPT_3+\frac{1}{2}\sqrt{\Delta\Delta^*}+\delta(\sqrt{\frac{\Delta}{2}}+\sqrt{\frac{\Delta^*}{2}})+4nb_{\max}^2+\delta^2,
\end{align}
where the first inequality is by $l^*_i,l_i\leq b_{\max}, \forall i$, the second inequality is by Cauchy-Schwartz inequality, and the last inequality is due to \eqref{eq:OPT_delta} and \eqref{eq:NE_delta}. Finally, by substituting $\Delta=2\gamma+2NE_3$ and $\Delta^*=2\gamma+2OPT_3$ in the last expression above, and simplifying the terms, we obtain the desired bound.
\end{proof}

\smallskip
\begin{lemma}\label{lemm:PoA-equation}
Let $\delta^2\!:=\frac{\sum_{j=1}^{m}g^2_j}{2}$ and $\gamma\!:=\delta^2+nb^2_{\max}$. If 
\begin{align}\label{eq:x-poa2}
x&\leq \sqrt{(\frac{\gamma}{n}+1)(\frac{\gamma}{n}+x)}+(1+\frac{\delta^2}{n}+4b_{\max}^2)\cr 
&\qquad+\frac{\delta}{\sqrt{n}}(\sqrt{\frac{\gamma}{n}+x}+\sqrt{\frac{\gamma}{n}+1}), 
\end{align}
we must have $x\leq 3+12b_{\max}^2+9\frac{\delta^2}{n}$.
\end{lemma}
\begin{proof}
Let $q:=1+\frac{\delta^2}{n}+4b_{\max}^2+\frac{\delta}{\sqrt{n}}\sqrt{\frac{\gamma}{n}+1}$, and $p:=\sqrt{\frac{\gamma}{n}+1}+\frac{\delta}{\sqrt{n}}$. Then we can rewrite \eqref{eq:x-poa2} as $x-q\leq p\sqrt{\frac{\gamma}{n}+x}$. Squaring both sides and solving for $x$ we obtain
\begin{align}\label{eq:p-q}
x\leq \frac{1}{2}(p^2+2q+p\sqrt{p^2+4q+4\frac{\gamma}{n}}).
\end{align}
As $p^2+4q+4\frac{\gamma}{n}\leq (p+\frac{2q}{p}+\frac{2\gamma}{pn})^2$, we can upper bound \eqref{eq:p-q} by $x\leq p^2+2q+\frac{\gamma}{n}$. Replacing the expressions for $p$ and $q$ into this relation and simplifying we obtain
\begin{align}\label{eq:x-gamma}
x\leq 3+8b_{\max}^2+3\frac{\delta^2}{n}+4\frac{\delta}{\sqrt{n}}\sqrt{1+\frac{\gamma}{n}}+2\frac{\gamma}{n}. 
\end{align}
Finally, using $\frac{\gamma}{n}=\frac{\delta^2}{n}+b^2_{\max}$ into \eqref{eq:x-gamma} and noting that $1+\frac{\gamma}{n}\leq (\frac{\delta}{\sqrt{n}}+\frac{b_{\max}^2\sqrt{n}}{2\delta})^2$, we obtain the desired result.
\end{proof} 

{\bf \subsection{Proof of Theorem \ref{thm:PoS}}}\label{appx:pos} 
To bound the PoS, we use the potential function method as in \cite[Theorem 19.13]{nisan2007algorithmic} to show that the social cost $C(\boldsymbol{a}):=\sum_{i}C_i(\boldsymbol{a})$ has nearly the same structure as the potential function $\Phi(\boldsymbol{a})$. Using the linear latency function $c_e(x)=a_e x+b_e$, and the quadratic energy pricing $f_j(x)=x^2, \forall j$ in the potential function \eqref{eq:potential}, we get
\begin{align}\label{eq:potential_linear_quadratic}
\Phi(\boldsymbol{a})&=\frac{1}{2}\sum_{e}\big(a_en_e^2+(a_e+2b_e)n_e\big)\!+\!\sum_{j=1}^{m}\frac{|\mathcal{Q}_j|(|\mathcal{Q}_j|\!+\!1)}{2\sigma_j}\cr 
&+\sum_{j=1}^{m}g_j^2+\sum_{j=1}^{m}L_j^2-2\sum_{j=1}^{m}g_jL_j+\sum_{i=1}^{n}\ln\Big(\frac{\bar{b}_i}{b_i+l_i}\Big),
\end{align} 
where $L_j:=\sum_{k\in \mathcal{Q}_j}l_k$. On the other hand, we have
\begin{align}\label{eq:social_cost}
C(\boldsymbol{a})&=\sum_{e}\big(a_en_e^2+b_en_e\big)+\sum_{j=1}^{m}\frac{|\mathcal{Q}_j|^2}{\sigma_j}-\sum_{j=1}^{m}\sum_{k\in \mathcal{Q}_j}l^2_k\cr 
&+2\sum_{j=1}^{m}L_j^2-2\sum_{j=1}^{m}g_jL_j+\sum_{i=1}^{n}\ln\Big(\frac{\bar{b}_{i}}{b_i+l_i}\Big).
\end{align}
Comparing \eqref{eq:potential_linear_quadratic} and \eqref{eq:social_cost}, we can write 
\begin{align}\nonumber
\frac{1}{2}C(\boldsymbol{a})\leq \Phi(\boldsymbol{a})&\leq C(\boldsymbol{a})+\sum_{j=1}^{m}\sum_{k\in \mathcal{Q}_j}l^2_k+\sum_{j=1}^{m}g_j^2.\cr 
&\leq C(\boldsymbol{a})+nb^2_{\max}+\sum_{j=1}^{m}g_j^2.
\end{align} 
Now let $\hat{\boldsymbol{a}}$ be the NE which minimizes the potential function $\Phi(\cdot)$, and $\boldsymbol{a}^*$ be the optimal action profile. Then,
\begin{align}\nonumber
C(\hat{\boldsymbol{a}})\leq 2\Phi(\hat{\boldsymbol{a}})\leq 2\Phi(\boldsymbol{a}^*)\leq 2[C(\boldsymbol{a}^*)+nb^2_{\max}+\sum_{j=1}^{m}g_j^2]. 
\end{align} 
Therefore, dividing both sides by $C(\boldsymbol{a}^*)\ge n$, we get 
\begin{align}\nonumber
\mbox{PoS}\leq\frac{C(\hat{\boldsymbol{a}})}{C(\boldsymbol{a}^*)}\leq 2\Big(1+b^2_{\max}+\frac{\sum_jg_j^2}{n}\Big).
\end{align}

{\bf \subsection{Proof of Theorem \ref{thm:chernoff}}}\label{appx:chernoff} 

Since $\{G_j, j\in\mathcal{M}\}$ are independent, so are their squares $\{G_j^2\}$, and we have $\mathbb{E}[\frac{\sum_{j=1}^{m}G_j^2}{m}]=\frac{\sum_{j=1}^{m}(\mu^2_j+\sigma^2_j)}{m}$. Using Hoeffding bound for independent and non-identical random variables we have 
\begin{align}\nonumber
\mathbb{P}\Big[\sum_{j=1}^{m}G_j^2-\sum_{j=1}^{m}(\mu^2_j+\sigma^2_j)>mt\Big]\leq \exp\big(-\frac{2mt^2}{K^2}\big).
\end{align}
Since $\mbox{PoA}\leq c+4.5(\frac{\sum_{j=1}^{m}G_j^2}{n})$, where $c=3+12b_{\max}^2$, by choosing $t=\frac{\frac{n}{4.5}-\sum_{j=1}^{m}(\mu_j^2+\sigma_j^2)}{m}$, we can write
\begin{align}\label{eq:prob-epsilon}
\mathbb{P}[\mbox{PoA}\ge c+1]&\leq \mathbb{P}\Big[\frac{\sum_{j=1}^{m}G_j^2}{n}>1\Big]\cr 
&= \mathbb{P}\Big[\frac{\sum_{j=1}^{m}G_j^2}{n}>\frac{\sum_{j=1}^{m}(\mu^2_j+\sigma^2_j)+mt}{n}\Big]\cr 
&\leq \exp\big(-\frac{2mt^2}{K^2}\big).  
\end{align}
Now in order the probability in \eqref{eq:prob-epsilon} to be less than $\epsilon$, we need to have $t\ge K\sqrt{\frac{\ln(\frac{1}{\epsilon})}{2m}}$. Finally, replacing the expression for $t$ in this inequality and solving for $n$, we obtain the desired bound.

\end{document}